\documentclass[%
reprint,
superscriptaddress,
 amsmath,amssymb,
 aps,
prb,
]{revtex4-2}

\allowdisplaybreaks 

\usepackage{graphicx}
\usepackage{dcolumn}
\usepackage[version=3]{mhchem} 
\usepackage{bm}
\usepackage{amssymb}
\usepackage{mathtools}
\usepackage[colorlinks=true,bookmarks=false,citecolor=blue,urlcolor=black,linkcolor=red]{hyperref}
\usepackage{booktabs}
\usepackage{setspace}
\usepackage{comment}

\usepackage{enumitem}

\usepackage{algorithm}
\usepackage{algpseudocode}

\usepackage{amsthm}

\newtheorem*{theorem*}{Theorem}
\newtheorem*{lemma*}{Lemma}

\newtheorem*{corollary*}{Corollary}
\newtheorem{proposition}{Proposition}

\theoremstyle{definition}

\theoremstyle{remark}

\newtheorem*{example*}{Example}
\newtheorem*{alg*}{Algorithm}

\newcommand{\lamd}[1]{\bm{\lambda}^{\downarrow}(#1)}
\newcommand{\lamu}[1]{\bm{\lambda}^{\uparrow}(#1)}

\begin{abstract}
We establish a majorization-based theory for bounding observables of waves with varied coherence. For any measurement, exact bounds are attained by the maximal and minimal elements in the set of input coherence spectra. The set's supremum and infimum---which may lie outside the set---provide optimal universal bounds: any alternative spectrum yielding universal bounds produces weaker constraints. We present an algorithm to compute the supremum and infimum, and prove that they lie either at singular boundary points or strictly outside the set of coherence spectra.
\end{abstract}

\begin{document}

\title{Optimal Universal Bounds for Waves with Varied Coherence \\ Based on Supremum and Infimum Coherence Spectra}

\author{Shiyu Li}
\affiliation{Department of Electrical and Computer Engineering and Microelectronics Research Center, The University of Texas at Austin, Austin, Texas 78712, USA}

\author{Cheng Guo}
\email{chengguo@utexas.edu}
\affiliation{Department of Electrical and Computer Engineering and Microelectronics Research Center, The University of Texas at Austin, Austin, Texas 78712, USA}

\date{\today}

\maketitle

\section{Introduction}\label{sec:intro}

Coherence is a fundamental property of all waves~\cite{born1999book,goodman2000book,mandel1995book,neill2003book,wolf2007}, and its control enables broad applications~\cite{glauber1963a,mandel1965,perina1985,korotkova2022}. High-coherence sources underpin interferometric metrology~\cite{Hariharan1991,udem2002,abbott2016}, holography~\cite{schnars1994,genevet2015,javidi2021}, and LiDAR~\cite{pierrottet2008,poulton2017}; low-coherence sources enable optical coherence tomography~\cite{huang1991,fercher2003} and speckle-free imaging~\cite{goodman1976,redding2012}; partially coherent sources benefit lithography~\cite{lin1980} and optical neural networks~\cite{luo2019,jia2024}. In practice, wave coherence often varies due to inherent or environmental fluctuations~\cite{born1999book,goodman2000book,mandel1995book,neill2003book}, and robust system design must accommodate such variation. This motivates a fundamental question: how does varied coherence affect achievable physical responses?

A common physical intuition is that lower coherence restricts the achievable range of responses. In Young's double-slit experiment, for example, less coherent sources yield reduced interference contrast~\cite{zernike1938}. This intuition suggests a natural strategy for bounding the effects of coherence variation: given a set of waves with different coherence, identify the most and least coherent elements, which should yield the maximal and minimal ranges of responses for any measurement.

\begin{figure}[tb]
\centering
\includegraphics[width=0.42\textwidth]{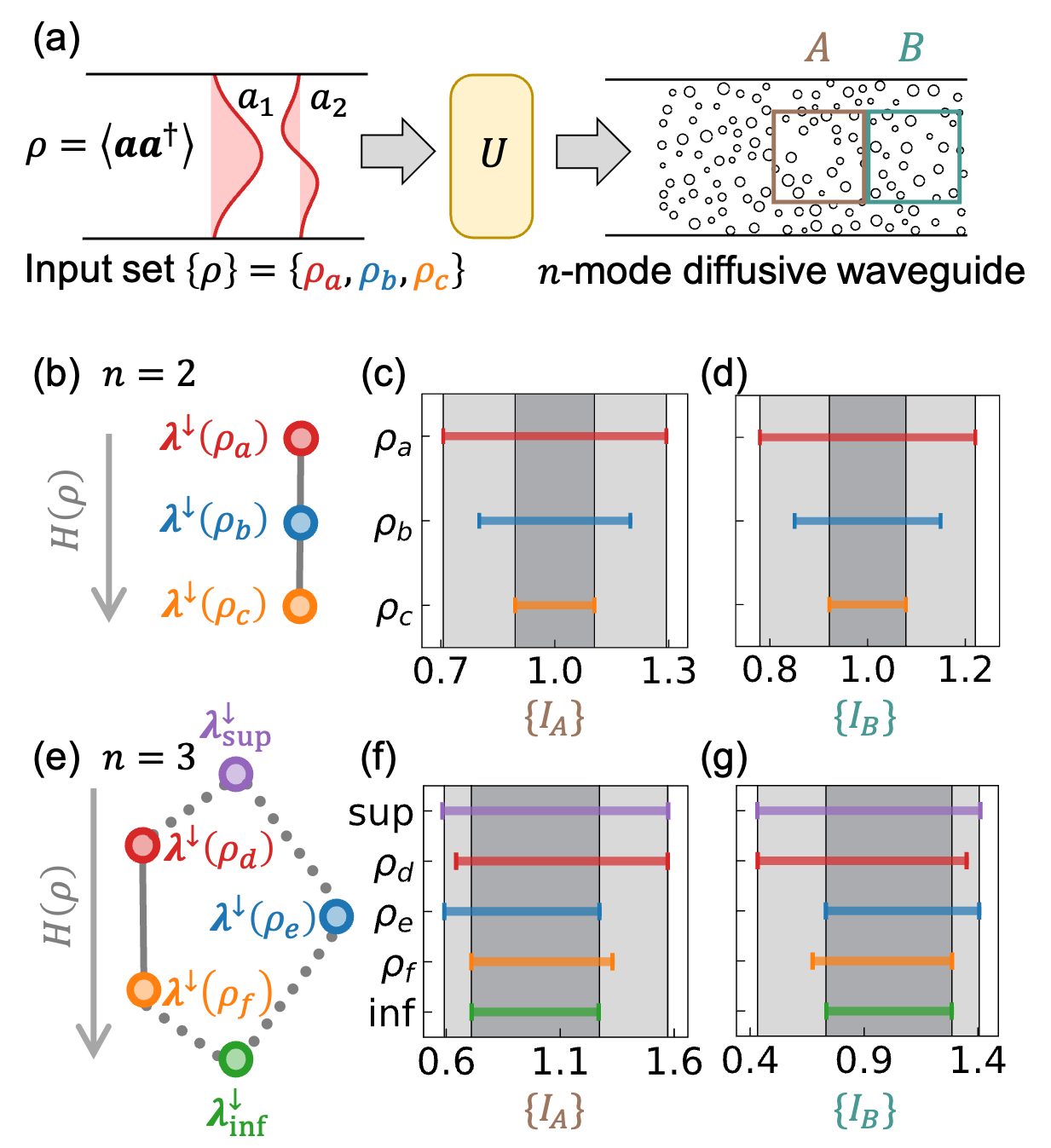}
\caption{Transport of waves with varied coherence. (a)~Partially coherent waves $\{\rho\}$, subject to unitary control $U$, are injected into an $n$-mode diffusive waveguide (illustrated for the case $n=2$). Intensities $I_A$ and $I_B$ are measured in regions $A$ and $B$. (b--d)~Results for $n=2$ with $\{\rho\} = \{\rho_a, \rho_b, \rho_c\}$. (b)~Hasse diagram: edges indicate majorization, and higher positions correspond to lower entropy $H(\rho)$. (c,d)~Intensity ranges satisfy $\{I\}_c \subseteq \{I\}_b \subseteq \{I\}_a$; light- and dark-shaded regions indicate the outer and inner bounds over $\{\rho\}$, attained by the maximum ($\rho_a$) and minimum ($\rho_c$) elements, respectively. (e--g)~Results for $n=3$ with $\{\rho\} = \{\rho_d, \rho_e, \rho_f\}$. (e)~Hasse diagram: the majorization order becomes partial. (f,g)~Intensity ranges no longer nest. The supremum $\bm{\lambda}^{\downarrow}_{\sup}$ and infimum $\bm{\lambda}^{\downarrow}_{\inf}$, lying outside the input set, provide optimal universal outer and inner bounds for all measurements.}
\label{fig:Fig1}
\end{figure}

While this strategy is intuitive, an immediate question arises: how should one compare wave coherence? The standard approach uses a scalar measure~\cite{karczewski1963a,tervo2003,refregier2005,setala2006} such as entropy~\cite{laue1907,gamo1964}. This works for two-mode waves but fails for waves with $n \geq 3$ modes. We illustrate this failure with an example in Fig.~\ref{fig:Fig1}(a) (simulation details in~Appendix~\ref{sec:simulation}). Consider an $n$-mode diffusive waveguide~\cite{cao2022a} into which we inject partially coherent waves and measure the intensities $I_A$ and $I_B$ in two target regions $A$ and $B$~\cite{shaughnessy2024multiregion}. The input waves have varied coherence, characterized by a set of density matrices $\{\rho\}$. For each $\rho$, we apply unitary control~\cite{guo2023b,guo2024,guo2024absorp,guo2024trans}: $\rho \to U\rho U^{\dagger}$, $U \in U(n)$, to generate all waves with the same power and coherence as $\rho$. We then determine the achievable ranges $\{I_A\}$ and $\{I_B\}$ for each $\rho$, along with the inner and outer bounds as $\rho$ varies over $\{\rho\}$.

Figures~\ref{fig:Fig1}(b--d) show results for an $n = 2$-mode waveguide with inputs $\{\rho\} = \{\rho_a, \rho_b, \rho_c\}$ (eigenvalues in~Appendix~\ref{sec:numerical_values}). Using the definition of entropy~\cite{laue1907,gamo1964}
\begin{equation}
H(\rho) = - \operatorname{tr} (\rho \ln \rho),
\end{equation}
we find $H(\rho_a) < H(\rho_b) < H(\rho_c)$ [Fig.~\ref{fig:Fig1}(b)]. The intensity ranges satisfy $\{I\}_c \subseteq \{I\}_b \subseteq \{I\}_a$ for both $I_A$ and $I_B$ measurements, and the inner and outer bounds over $\{\rho\}$ (dark and light shades) are attained by $\rho_c$ and $\rho_a$, respectively. The entropy-based strategy succeeds.

In contrast, Figs.~\ref{fig:Fig1}(e--g) show results for an $n = 3$-mode waveguide with inputs $\{\rho\} = \{\rho_d, \rho_e, \rho_f\}$ (eigenvalues in~Appendix~\ref{sec:numerical_values}). Despite $H(\rho_d) < H(\rho_e) < H(\rho_f)$ [Fig.~\ref{fig:Fig1}(e)], the intensity ranges no longer nest: $\{I\}_f \nsubseteq \{I\}_e \nsubseteq \{I\}_d$ for both $I_A$ and $I_B$, and the bounds are not attained by $\rho_f$ and $\rho_d$. The entropy-based strategy fails.

This example reveals a fundamental gap in our understanding: scalar measures such as entropy fail to capture the coherence-observable relationship for multidimensional waves. This failure motivates a different approach. Recent work has introduced a vector order based on majorization to compare wave coherence~\cite{nielsen1999,gour2015,luis2016,bengtsson2017,gour2018}. This order has been shown to play a fundamental role in determining how coherence affects observables~\cite{guo2025transport}. Here we apply the majorization order to address the challenge of bounding observables under coherence variation.

In this paper, we establish a majorization-based theory for bounding observables of waves with varied coherence. The distinction between the $n = 2$ and $n \geq 3$ cases arises from the partial-order nature of majorization: for a set $\Lambda$ of coherence spectra, maximal and minimal elements are generally non-unique, precluding a single most or least coherent element within $\Lambda$. Nevertheless, thanks to the complete-lattice property of the majorization order, there exist unique supremum and infimum spectra---the least upper and greatest lower bounds of $\Lambda$---which may lie outside $\Lambda$. We prove that these spectra provide optimal universal bounds for all measurements. We develop an algorithm to compute the supremum and infimum for any compact set $\Lambda$, and classify their geometric locations based on the smoothness of the boundary. Our results provide a theoretical foundation for understanding how coherence variation constrains wave phenomena.

\section{Physical problem}\label{sec:problem}

Consider an $n$-dimensional wave represented by an $n \times n$ density matrix $\rho$~\cite{goodman2000book,mandel1995book,wolf1985,yamazoe2012,okoro2017,saleh2025book}. The matrix $\rho$ is positive semidefinite, Hermitian, and normalized to unit trace: $\operatorname{tr} \rho = 1$~\cite{saleh2025book}. The eigenvalues of $\rho$, termed the coherence spectrum, are given by
\begin{equation}
    \lamd{\rho} = (\lambda_1^{\downarrow}(\rho), \ldots, \lambda_n^{\downarrow}(\rho)), \label{eq:coherence_spectrum}
\end{equation}
where $\downarrow$ denotes ordering the components in non-increasing order. The spectrum $\lamd{\rho}$ encodes the coherence properties of the wave. All feasible $\lamd{\rho}$ constitute the set of ordered $n$-dimensional probability vectors:
\begin{equation}
    \Delta^{\downarrow}_n = \left\{\bm{x} \in \mathbb{R}^n \,\big|\, x_i \geq 0, \, x_i \geq x_{i+1}, \, \sum_{i=1}^{n} x_i = 1\right\}. \label{eq:probability_simplex}
\end{equation}

The set of all waves with the same power and coherence spectrum as $\rho$ can be realized through unitary control~\cite{guo2023b,guo2024,guo2024absorp,guo2024trans}:
\begin{equation}
    \{\rho\} \coloneqq \{U \rho U^{\dagger} \mid U \in U(n)\}, \label{eq:unitary_orbit}
\end{equation}
which can be experimentally implemented using spatial light modulators~\cite{vellekoop2007,popoff2014,yu2017}, Mach-Zehnder interferometers~\cite{reck1994,miller2013,miller2013self,miller2013self-config,carolan2015,miller2015,clements2016,ribeiro2016,wilkes2016,annoni2017,perez2017,miller2017,harris2018,pai2019}, and multiplane light conversion systems~\cite{morizur2010,labroille2014,kupianskyi2023,taguchi2023,zhang2023}.
For an input wave $\rho$, a measurement characterized by a Hermitian operator $O$ yields an observable value:
\begin{equation}
    o[\rho] = \operatorname{tr} \rho O. \label{eq:measurement}
\end{equation}
The range of achievable values for the whole set of inputs $\{\rho\}$ is a closed interval~\cite{guo2024absorp,guo2024trans,guo2025transport}:
\begin{equation}
    \{o\} = \left[\lamd{\rho} \cdot \lamu{O}, \, \lamd{\rho} \cdot \lamd{O}\right], \label{eq:k_range}
\end{equation}
where $\cdot$ denotes the inner product.

Here we consider waves with varied coherence whose coherence spectra take values from a feasible set
\begin{equation}
    \Lambda \subseteq \Delta^{\downarrow}_n. \label{eq:coherence_set}
\end{equation}
Our aim is to determine two bounds:

(1) The \textit{outer bounds}: the union of achievable ranges across all coherence spectra
\begin{equation}\label{eq:def_outer_bounds}
    \bigcup_{k \in \Lambda} \{o\}_k = \left[\lamd{\rho_{ol}} \cdot \lamu{O}, \, \lamd{\rho_{ou}} \cdot \lamd{O}\right],
\end{equation}
as illustrated by the light-shaded regions in Figs.~\ref{fig:Fig1}(c,d,f,g). Here $\lamd{\rho_{ol}}$ and $\lamd{\rho_{ou}}$ denote coherence spectra that attain the lower and upper outer bounds, respectively. Physically, the outer bounds represent the range of all possible observable values across all input coherence spectra.

(2) The \textit{inner bounds}: the intersection of achievable ranges across all coherence spectra
\begin{equation}\label{eq:def_inner_bounds}
    \bigcap_{k \in \Lambda} \{o\}_k = \left[\lamd{\rho_{il}} \cdot \lamu{O}, \, \lamd{\rho_{iu}} \cdot \lamd{O}\right],
\end{equation}
as illustrated by the dark-shaded regions in Figs.~\ref{fig:Fig1}(c,d,f,g). Here $\lamd{\rho_{il}}$ and $\lamd{\rho_{iu}}$ denote coherence spectra that attain the lower and upper inner bounds, respectively. Physically, the inner bounds represent the guaranteed range of observable values that are attainable for any input coherence spectrum in $\Lambda$.

Given $O$ and $\Lambda$, in principle one can calculate the outer and inner bounds using Eqs.~(\ref{eq:k_range}), (\ref{eq:def_outer_bounds}), and (\ref{eq:def_inner_bounds}). However, such direct calculation faces challenges. First, it requires optimizing for $\lamd{\rho_{ol}}$, $\lamd{\rho_{ou}}$, $\lamd{\rho_{il}}$, and $\lamd{\rho_{iu}}$ over $\Lambda$, which becomes intractable when $\Lambda$ is large or infinite. Second, the optimized spectra depend on $O$, requiring the optimization to be repeated for each different measurement. These difficulties motivate the search for simpler bounds that are readily calculated for any measurement.

\section{Majorization}\label{sec:majorization}
To address this challenge, we use a coherence order based on majorization. For vectors $\bm{x}$ and $\bm{y}$ in $\Delta_n^\downarrow$, $\bm{x}$ is majorized by $\bm{y}$, denoted as $\bm{x} \prec \bm{y}$~\cite{marshall2011}, if 
\begin{align}
    \sum_{i=1}^k x_i^{\downarrow} \leq \sum_{i=1}^k y_i^{\downarrow}, \quad \text{for all } k = 1,2,\dots,n-1. \label{eq:majorization_1} 
\end{align}
The majorization relation $\prec$ defines an order on $\Delta_n^\downarrow$~\cite{marshall2011}. If $\bm{\lambda}^\downarrow(\rho_1) \prec \bm{\lambda}^\downarrow(\rho_2)$, we say that $\rho_1$ is no more coherent than $\rho_2$ in the majorization order~\cite{guo2025transport}. Unlike the entropy order, which is a total order, the majorization order is a partial order when $n\geq 3$~\cite{davey2002book,cicalese2002}. This means that $\bm{x}$ and $\bm{y}$ in $\Delta_n^\downarrow$ can be incomparable, denoted as $\bm{x} \| \bm{y}$, when neither $\bm{x} \prec \bm{y}$ nor $\bm{y} \prec \bm{x}$ holds~\cite{davey2002book}. Incomparability is typical rather than exceptional~\cite{cunden2021,jain2024,harling2024a}. 

Within a subset $\Lambda \subseteq \Delta_n^\downarrow$, we define four types of special elements~\cite{schroder2003,boyd2004,roman2008}. A maximal element is not majorized by any other element of $\Lambda$. A minimal element does not majorize any other element. A maximum element majorizes every element. A minimum element is majorized by every element. The sets of all maximal and minimal elements are denoted by $\Lambda_{\max}$ and $\Lambda_{\min}$, respectively. If the maximum (minimum) of $\Lambda$ exists, it is unique and $\Lambda_{\max}$ ($\Lambda_{\min}$) becomes a singleton~\cite{davey2002book,boyd2004}.

The majorization order plays a fundamental role in the coherence theory of observables. It is proved that~\cite{guo2025transport}
\begin{equation}\label{eq:majorization_order_preserved}
    \bm{\lambda}^\downarrow(\rho_1) \prec \bm{\lambda}^\downarrow(\rho_2) \iff \forall O\in H_n, \; \{o\}_1 \subseteq \{o\}_2, 
\end{equation}
where $H_n$ denotes the set of all $n\times n$ Hermitian matrices. Eq.~(\ref{eq:majorization_order_preserved}) indicates that the majorization order is preserved and reflected in the range of any measurement.

We use the majorization order to explain the results in Fig.~\ref{fig:Fig1}. Figures~\ref{fig:Fig1}(b) and~\ref{fig:Fig1}(e) show the Hasse diagrams~\cite{davey2002book} for the majorization order for the $n=2$ and $n=3$ cases. When $n=2$, the majorization order is total:
\begin{equation}
\lamd{\rho_c} \prec \lamd{\rho_b} \prec \lamd{\rho_a}.
\end{equation}
Thus, $\lamd{\rho_a}$ and $\lamd{\rho_c}$ are the unique maximum and minimum, respectively. They attain the outer and inner bounds for any measurement. This explains
\begin{equation}
    \{I\}_c \subseteq \{I\}_b \subseteq \{I\}_a
\end{equation}
for both $I_A$ and $I_B$ measurements [Figs.~\ref{fig:Fig1}(c) and~\ref{fig:Fig1}(d)].

When $n=3$, the majorization order becomes partial:
\begin{equation}
\lamd{\rho_f} \prec \lamd{\rho_d}, \quad  \lamd{\rho_f} \parallel \lamd{\rho_e}, \quad \lamd{\rho_d} \parallel \lamd{\rho_e}.   
\end{equation}
Thus, the maximal and minimal elements are not unique:
\begin{align}
\Lambda_{\max} &= \{\lamd{\rho_d}, \lamd{\rho_e}\},\\ \quad \Lambda_{\min} &= \{\lamd{\rho_f}, \lamd{\rho_e}\}.    
\end{align}
For the $I_A$ measurement, we have 
\begin{align}
\lamd{\rho_{ol}}= \lamd{\rho_e}, \quad \lamd{\rho_{ou}}= \lamd{\rho_d}, \\
\lamd{\rho_{il}}= \lamd{\rho_f}, \quad 
\lamd{\rho_{iu}}=\lamd{\rho_e}.
\end{align}
For the $I_B$ measurement, we have 
\begin{align}
\lamd{\rho_{ol}}= \lamd{\rho_d}, \quad \lamd{\rho_{ou}}= \lamd{\rho_e}, \\
\lamd{\rho_{il}}= \lamd{\rho_e}, \quad 
\lamd{\rho_{iu}}=\lamd{\rho_f}.
\end{align}
We observe that the outer (inner) bounds are attained by maximal (minimal) elements, but the precise choices depend on the measurement. This non-uniqueness of maximal and minimal elements, arising from incomparability in the majorization order, precludes a single fixed element in $\Lambda$ from establishing outer or inner bounds for all measurements universally. Such phenomena cannot be explained by total orders such as the entropy order.

\begin{figure}[tb]
\centering
\includegraphics[width=0.43\textwidth]{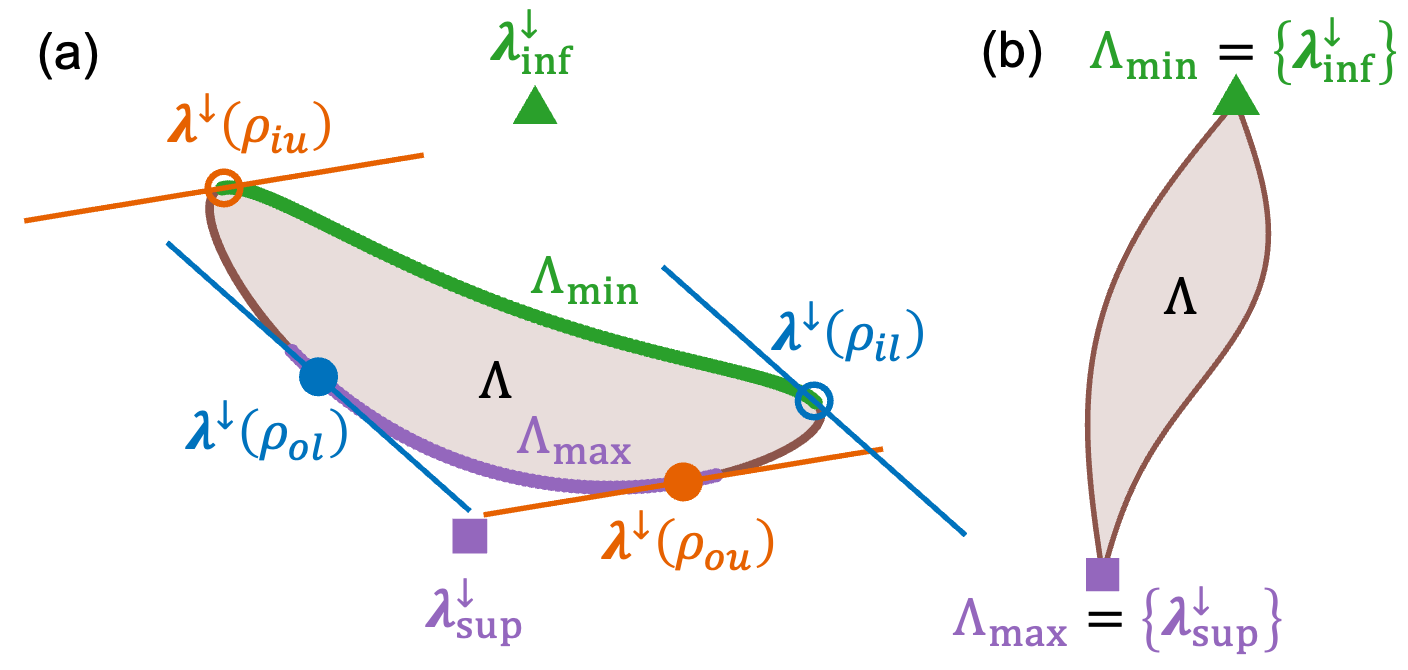}
\caption{Schematic of a compact set $\Lambda$ of coherence spectra. (a)~Generic case: $\Lambda$ has sets of maximal and minimal elements $\Lambda_{\max}$ and $\Lambda_{\min}$. Outer bounds are attained by elements of $\Lambda_{\max}$, and inner bounds by elements of $\Lambda_{\min}$; the specific elements depend on the measurement $O$. The supremum $\bm{\lambda}^{\downarrow}_{\sup}$ and infimum $\bm{\lambda}^{\downarrow}_{\inf}$ lie outside $\Lambda$. (b)~Special case: $\bm{\lambda}^{\downarrow}_{\sup}$ and $\bm{\lambda}^{\downarrow}_{\inf}$ coincide with the unique maximum and minimum of $\Lambda$.}
\label{fig:Fig2}
\end{figure}

We generalize the results above to any compact set $\Lambda$ (finite or closed infinite). We prove that for any $\bm{x} \in \Lambda$, there exist $\bm{x}_m \in \Lambda_{\min}$ and $\bm{x}_M \in \Lambda_{\max}$ such that $\bm{x}_m \prec \bm{x} \prec \bm{x}_M$. Consequently, to obtain the outer bounds over $\Lambda$, it suffices to optimize over $\Lambda_{\max}$: 
\begin{equation}\label{eq:restrict_outer_bounds}
\forall O\in H_n, \quad    \bigcup_{k \in \Lambda} \{o\}_k = \bigcup_{k \in \Lambda_{\max}} \{o\}_k.
\end{equation}
For the inner bounds, it suffices to optimize over $\Lambda_{\min}$:
\begin{equation}\label{eq:restrict_inner_bounds}
\forall O\in H_n, \quad     \bigcap_{k \in \Lambda} \{o\}_k = \bigcap_{k \in \Lambda_{\min}} \{o\}_k.
\end{equation}
See~Appendix~\ref{sec:proof_maxmin} for proofs of Eqs.~(\ref{eq:restrict_outer_bounds}) and~(\ref{eq:restrict_inner_bounds}). These results admit a geometric interpretation when $\Lambda$ is a closed infinite set [Fig.~\ref{fig:Fig2}(a)]. $\lamd{\rho_{ou}}$ and $\lamd{\rho_{iu}}$ maximize and minimize $\lamd{O} \cdot \bm{x}$ over all $\bm{x} \in \Lambda$; they correspond to the tangent points where hyperplanes orthogonal to $\lamd{O}$ (orange lines) touch the boundary of $\Lambda$. Similarly, $\lamd{\rho_{il}}$ and $\lamd{\rho_{ol}}$ maximize and minimize $\lamu{O} \cdot \bm{x}$ over all $\bm{x} \in \Lambda$; they are the tangent points for hyperplanes orthogonal to $\lamu{O}$ (blue lines). These tangent points lie within restricted regions on the boundary of $\Lambda$: $\lamd{\rho_{ou}}$ and $\lamd{\rho_{ol}}$ lie within $\Lambda_{\max}$, while $\lamd{\rho_{iu}}$ and $\lamd{\rho_{il}}$ lie within $\Lambda_{\min}$. As the measurement $O$ varies, the tangent points slide across their respective regions. Only when $\Lambda_{\max}$ (or $\Lambda_{\min}$) is a singleton does it establish universal outer (inner) bounds for all measurements [Fig.~\ref{fig:Fig2}(b)].

\section{Supremum and infimum}\label{sec:sup_inf}
Although we have simplified the problem, we still need to optimize over $\Lambda_{\max}$ and $\Lambda_{\min}$, and the optimal spectra vary with the measurement. Ideally, a single fixed element in $\Lambda$ would establish outer (or inner) bounds universally for all measurements. This occurs when $\Lambda$ contains the maximum (or minimum) element, but a generic $\Lambda$ does not. To overcome this limitation, we relax the search domain from $\Lambda$ to $\Delta_n^\downarrow$: we seek a single fixed element in $\Delta_n^\downarrow$ that establishes \emph{optimal} outer (inner) bounds \emph{universally} for all measurements.

Such elements exist for both outer and inner bounds. The partially ordered set $(\Delta^{\downarrow}_n, \prec)$ forms a \emph{complete lattice}, which guarantees that any subset $\Lambda \subseteq \Delta_n^\downarrow$ has a unique supremum $\bm{\lambda}^{\downarrow}_{\sup}$ and infimum $\bm{\lambda}^{\downarrow}_{\inf}$ within $\Delta^{\downarrow}_n$~\cite{alberti1982,bapat1991,bosyk2019optimal} [Fig.~\ref{fig:Fig2}(a)]. These extrema are defined by~\cite{davey2002book,roman2008}  
\begin{align}
\Lambda \prec \lamd{\rho} &\iff \bm{\lambda}^{\downarrow}_{\sup} \prec \lamd{\rho}, \label{eq:def_sup}\\
\lamd{\rho} \prec \Lambda &\iff \lamd{\rho} \prec \bm{\lambda}^{\downarrow}_{\inf}, \label{eq:def_inf} 
\end{align}
where $\Lambda \prec \lamd{\rho}$ means that $\lamd{\rho}$ majorizes all elements of $\Lambda$. Thus, $\bm{\lambda}^{\downarrow}_{\sup}$ represents the least coherent spectrum that majorizes all elements of $\Lambda$, while $\bm{\lambda}^{\downarrow}_{\inf}$ represents the most coherent spectrum majorized by all elements of $\Lambda$. 

For a measurement $O$, we define the supremum and infimum ranges as the closed intervals
\begin{align}
    \{o\}_{\sup} &\coloneqq \bigl[\bm{\lambda}^{\downarrow}_{\sup} \cdot \lamu{O}, \bm{\lambda}^{\downarrow}_{\sup} \cdot \lamd{O} \bigr], \label{eq:sup_range} \\
    \{o\}_{\inf} &\coloneqq \bigl[\bm{\lambda}^{\downarrow}_{\inf} \cdot \lamu{O}, \bm{\lambda}^{\downarrow}_{\inf} \cdot \lamd{O} \bigr]. \label{eq:inf_range} 
\end{align}
We prove that they provide universal bounds:
\begin{equation}
\forall O\in H_n, \; \; \{ o \}_{\inf}\subseteq \bigcap_{k \in \Lambda}\{o\}_k \subseteq \bigcup_{k \in \Lambda}\{o\}_k \subseteq \{ o \}_{\sup}.
    \label{eq:set_relation}     
\end{equation}
\begin{proof}
The definition of $\bm{\lambda}^{\downarrow}_{\sup}$ [Eq.~(\ref{eq:def_sup})] and Eq.~(\ref{eq:majorization_order_preserved}) imply that $\forall k \in \Lambda$, $\{o\}_{k} \subseteq  \{o\}_{\sup}$, and thus $\bigcup_{k \in \Lambda} \{o\}_{k} \subseteq \{o\}_{\sup}$. Similarly, the definition of $\bm{\lambda}^{\downarrow}_{\inf}$ [Eq.~(\ref{eq:def_inf})] and Eq.~(\ref{eq:majorization_order_preserved}) imply that $\forall k \in \Lambda$, $\{o\}_{\inf} \subseteq \{o\}_{k}$, and thus $\{o\}_{\inf} \subseteq \bigcap_{k \in \Lambda} \{o\}_{k}$. The inclusion $\bigcap_{k \in \Lambda} \{o\}_{k} \subseteq \bigcup_{k \in \Lambda} \{o\}_{k}$ holds trivially. This completes the proof of Eq.~(\ref{eq:set_relation}).
\end{proof}

Moreover, these bounds are optimal. Any spectrum $\bm{\lambda}^\downarrow_u$ providing a universal outer bound must satisfy $\bm{\lambda}^{\downarrow}_{\sup} \prec \bm{\lambda}^\downarrow_u$, which implies
\begin{equation}
\forall O\in H_n, \quad \{o\}_{\sup}\subseteq\{o\}_u.\label{eq:optimality_sup}
\end{equation}
Similarly, any spectrum $\bm{\lambda}^\downarrow_l$ providing a universal inner bound must satisfy $\bm{\lambda}^\downarrow_l \prec \bm{\lambda}^{\downarrow}_{\inf}$, which implies
\begin{equation}
\forall O\in H_n, \quad \{o\}_l\subseteq\{o\}_{\inf}.\label{eq:optimality_inf}
\end{equation}
Thus, $\bm{\lambda}^{\downarrow}_{\sup}$ and $\bm{\lambda}^{\downarrow}_{\inf}$ provide the tightest universal outer and inner bounds constructed from a single spectrum.
\begin{proof}
We first prove Eq.~(\ref{eq:optimality_sup}). By premise, $\bm{\lambda}^{\downarrow}_u$ satisfies 
$\forall O\in H_n,\, \bigcup_{k \in \Lambda} \{ o \}_k \subseteq \{ o \}_{u}$.
From Eq.~(\ref{eq:majorization_order_preserved}), we obtain $\Lambda \prec \bm{\lambda}^\downarrow_u$. Then Eq.~(\ref{eq:def_sup}) implies $\bm{\lambda}^\downarrow_{\sup} \prec \bm{\lambda}^\downarrow_u$. Using Eq.~(\ref{eq:majorization_order_preserved}) again, we obtain $\forall O\in H_n,\, \{o\}_{\sup}\subseteq\{o\}_u.$ This completes the proof of Eq.~(\ref{eq:optimality_sup}). 

The proof of Eq.~(\ref{eq:optimality_inf}) is similar: $\bm{\lambda}^{\downarrow}_l$ satisfies $\forall O\in H_n$, $\{ o \}_l \subseteq \bigcap_{k \in \Lambda} \{ o \}_k$, so $\bm{\lambda}^\downarrow_l \prec \Lambda$ [Eq.~(\ref{eq:majorization_order_preserved})]. Consequently, $\bm{\lambda}^\downarrow_l \prec \bm{\lambda}^\downarrow_{\inf}$ [Eq.~(\ref{eq:def_inf})], which implies $\{o\}_{l} \subseteq \{o\}_{\inf}$ [Eq.~(\ref{eq:majorization_order_preserved})].
This completes the proof of Eq.~(\ref{eq:optimality_inf}).
\end{proof}

\section{Algorithm}\label{sec:algorithm}
Having established the significance of $\bm{\lambda}^{\downarrow}_{\sup}$ and $\bm{\lambda}^{\downarrow}_{\inf}$, we now address their efficient computation.

For a finite set $\Lambda$, we compute $\bm{\lambda}^{\downarrow}_{\sup}$ and $\bm{\lambda}^{\downarrow}_{\inf}$ using the algorithm in Ref.~\cite{bosyk2019optimal}. For the set $\{\lamd{\rho_d}, \lamd{\rho_e}, \lamd{\rho_f}\}$ in Fig.~\ref{fig:Fig1}(e), we obtain 
\begin{equation}
\bm{\lambda}^{\downarrow}_{\sup} = (0.80, 0.20, 0.00), \quad \bm{\lambda}^{\downarrow}_{\inf} = (0.55, 0.35, 0.10).    
\end{equation}
Figs.~\ref{fig:Fig1}(f,g) confirm that for both $I_A$ and $I_B$,
\begin{equation}
    \{ I \}_{\inf}\subseteq \{I\}_k  \subseteq \{ I \}_{\sup}, \quad k \in \{d,e,f\}.
\end{equation}
The supremum and infimum bounds (purple and green intervals) closely match the exact outer and inner bounds (light and dark shades):
\begin{align}
\{I_A\}_{\sup} &= [0.58, 1.57],  &\{I_A\}_{\text{outer}} = [0.59, 1.57]; \\
\{I_A\}_{\inf} &= [0.71, 1.27],  &\{I_A\}_{\text{inner}} = [0.71, 1.27].  \\
\{I_B\}_{\sup} &= [0.43, 1.41],  &\{I_B\}_{\text{outer}} = [0.43, 1.40]; \\
\{I_B\}_{\inf} &= [0.73, 1.29],  &\{I_B\}_{\text{inner}} = [0.73, 1.29].  
\end{align}

\begin{figure}[tb]
\centering
\includegraphics[width=0.42\textwidth]{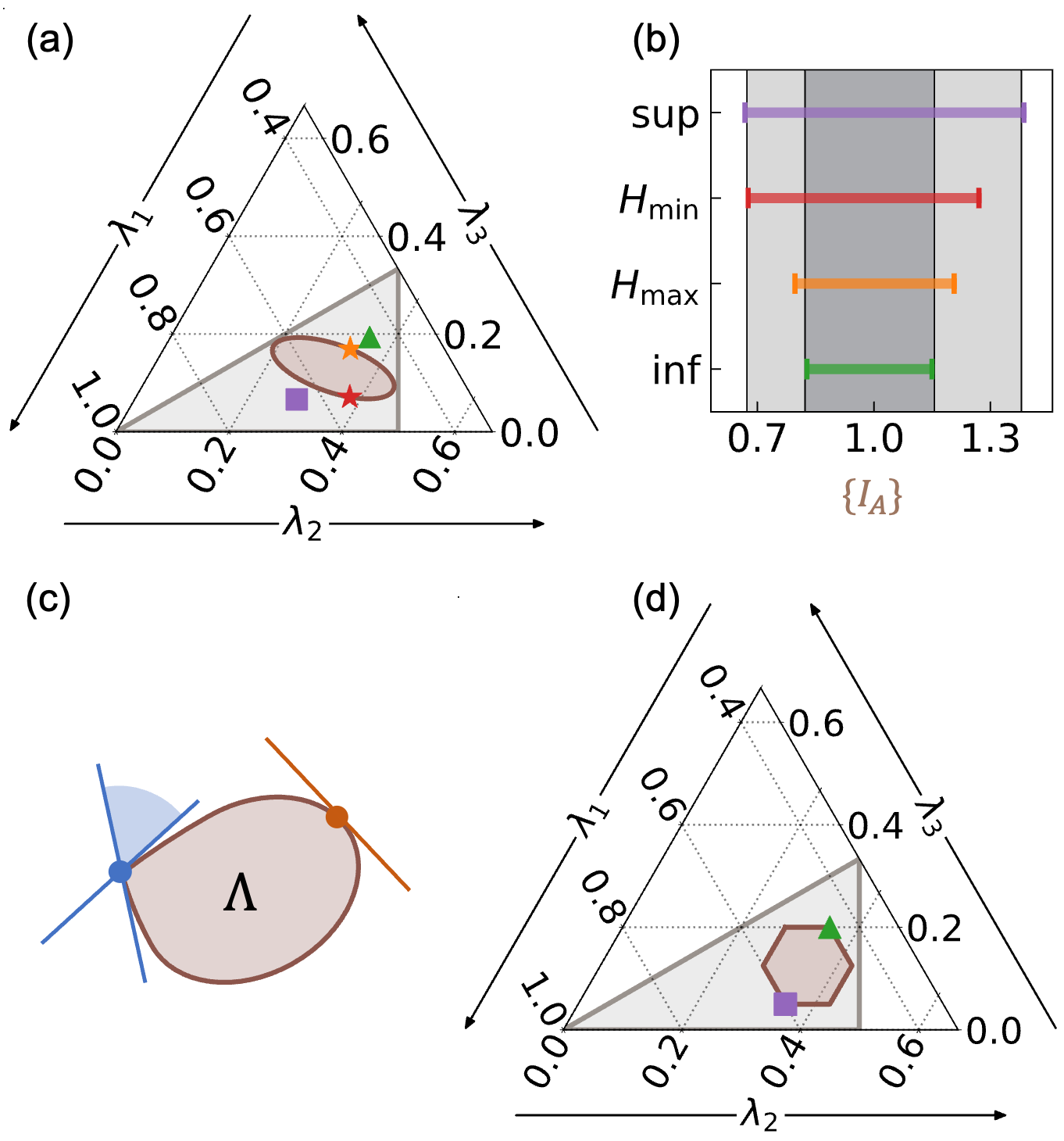}
\caption{%
Supremum and infimum of a compact set $\Lambda \subseteq \Delta_n^\downarrow$. 
(a)~Ternary plot of an elliptical disk $\Lambda$ (brown) in $\Delta^\downarrow_3$ (gray). Purple square and green triangle mark the supremum and infimum; red and orange stars mark the entropy minimum and maximum. 
(b)~Achievable range $\{I_A\}$ for the $n=3$ waveguide in Fig.~\ref{fig:Fig1}(a). The supremum and infimum bounds closely match the exact bounds (light and dark shades), whereas the entropy extrema do not. Results for $\{I_B\}$ are similar (Fig.~\ref{fig:Fig8}). 
(c)~Smooth (orange) and singular (blue) boundary points of a convex set. The supremum and infimum can lie only at singular points or outside $\Lambda$. 
(d)~A hexagonal $\Lambda$; the supremum and infimum lie at vertices (singular points).}
\label{fig:Fig3}
\end{figure}

For a closed infinite set $\Lambda$, we develop an algorithm to compute $\bm{\lambda}^{\downarrow}_{\sup}$ and $\bm{\lambda}^{\downarrow}_{\inf}$.
\begin{alg*}[Computing $\bm{\lambda}^{\downarrow}_{\sup}$ and $\bm{\lambda}^{\downarrow}_{\inf}$ of compact $\Lambda$] \hfill
\begin{enumerate}[leftmargin=*, label=(\arabic*)]\setlength{\itemsep}{-0.1pt}
    \item Calculate the convex hull of $\Lambda$.
    \item If the convex hull is a polytope: $\bm{\lambda}^{\downarrow}_{\sup}$ and $\bm{\lambda}^{\downarrow}_{\inf}$ of $\Lambda$ coincide with those of the polytope's vertices, which can be calculated using the algorithm in Ref.~\cite{bosyk2019optimal}.
    
    \item If the convex hull is not a polytope: 

    \begin{enumerate}[leftmargin=1em,label=(\alph*)]\setlength{\itemsep}{-0.1pt}
        \item Calculate inscribing ($\mathcal{I}$) and circumscribing ($\mathcal{C}$) polytopes of the convex hull with $N$ vertices.
        \item Calculate the supremum and infimum of $\mathcal{I}$ and $\mathcal{C}$: $\sup(\mathcal{I})$, $\sup(\mathcal{C})$, $\inf(\mathcal{I})$, and $\inf(\mathcal{C})$. Increase $N$ until $\| \sup(\mathcal{I}) - \sup(\mathcal{C}) \|_{\infty} < \epsilon$ and $\| \inf(\mathcal{I}) - \inf(\mathcal{C}) \|_{\infty} < \epsilon$ with a predetermined tolerance $\epsilon$.
        \item Once the process converges, $\bm{\lambda}^{\downarrow}_{\sup}$ and $\bm{\lambda}^{\downarrow}_{\inf}$ of $\Lambda$ are approximated by $\sup(\mathcal{C})$ and $\inf(\mathcal{C})$, respectively.
    \end{enumerate}
\end{enumerate}
\end{alg*}
\noindent See~Appendix~\ref{sec:inscribing_circumscribing} for the construction of inscribing and circumscribing polygons. See~Appendix~\ref{sec:convergence} for the convergence analysis, which follows a trend of $\mathcal{O}(N^{-2})$. We adopt $\epsilon = 0.01$ throughout the paper.

We demonstrate the algorithm using an elliptic-disk set $\Lambda \subseteq \Delta_3^\downarrow$ [Fig.~\ref{fig:Fig3}(a)].
The resulting $\bm{\lambda}^{\downarrow}_{\sup}$ and $\bm{\lambda}^{\downarrow}_{\inf}$ lie outside $\Lambda$ [Fig.~\ref{fig:Fig3}(b)]. Again, the supremum and infimum bounds closely match the exact outer and inner bounds:
\begin{align}
\{I_A\}_{\sup} &= [0.67, 1.39],  &\{I_A\}_{\text{outer}} = [0.67, 1.38]; \\
\{I_A\}_{\inf} &= [0.83, 1.15],  &\{I_A\}_{\text{inner}} = [0.82, 1.16].     
\end{align} In contrast, the coherence spectra with minimum ($H_{\min}$) and maximum ($H_{\max}$) entropy fail to bound either measurement. Results for $\{ I_B \}$ are similar (Fig.~\ref{fig:Fig8}).

Additional numerical demonstrations for supremum and infimum bounds are provided in~Appendix~\ref{sec:additional_ND}.

\section{Locations of \\ supremum and infimum}\label{sec:geometric_loc}
Finally, we characterize the geometric locations of $\bm{\lambda}^{\downarrow}_{\sup}$ and $\bm{\lambda}^{\downarrow}_{\inf}$.

For definiteness, we assume that $\Lambda$ is a compact convex set with nonempty interior. Every boundary point $\bm{x}$ of $\Lambda$ admits a supporting hyperplane~\cite{gruber2007}. If this hyperplane is unique, $\bm{x}$ is called a \emph{smooth} boundary point; otherwise, it is called a \emph{singular} boundary point [Fig.~\ref{fig:Fig3}(c)]. If all boundary points are smooth, $\Lambda$ itself is said to be smooth.

Figures~\ref{fig:Fig2} and~\ref{fig:Fig3}(a) suggest that the locations of $\bm{\lambda}^{\downarrow}_{\sup}$ and $\bm{\lambda}^{\downarrow}_{\inf}$ depend on the boundary smoothness. We prove this: $\bm{\lambda}^{\downarrow}_{\sup}$ and $\bm{\lambda}^{\downarrow}_{\inf}$ of $\Lambda$ lie either at singular boundary points [Fig.~\ref{fig:Fig3}(d)] or outside $\Lambda$. If $\Lambda$ is smooth, both lie outside $\Lambda$. See~Appendix~\ref{sec:proof_location} for a proof.

\section{Conclusion}
We have established a majorization-based theory for bounding observables of waves with varied coherence. We prove that the exact bounds on any given measurement are attained by the maximal and minimal elements of the input coherence set. The supremum and infimum of the set provide optimal universal bounds for all measurements---any alternative spectrum yielding universal bounds must produce weaker constraints. For closed infinite sets, we developed a convergent algorithm to compute these spectra and classified their geometric locations: they lie either at singular boundary points or strictly outside the set of input coherence spectra. Our results apply to arbitrary wave types and any linear measurements, revealing fundamental constraints that coherence variation imposes on achievable physical responses.

\section*{Acknowledgments}

C.G. is supported by the Jack Kilby/Texas Instruments Endowed Faculty Fellowship.

\bibliography{main}

\clearpage

\twocolumngrid
\appendix
\onecolumngrid

\section{Simulation of power delivery in a diffusive waveguide}
\label{sec:simulation}

In this section, we provide computational details for the power delivery simulations.

As illustrated in Fig.~\ref{fig:Fig4}, we consider two silicon waveguides (refractive index $n_i = 3.48$) embedded in air. The narrower waveguide has a width of $0.4\,\mu$m and supports $n = 2$ eigenmodes; the wider waveguide is $0.6\,\mu$m wide and supports $n = 3$ eigenmodes. We simulate light propagation at vacuum wavelength $\lambda_0 = 1.55\,\mu$m for the transverse-magnetic (TM) polarization, where the electric field is polarized along the $z$-direction.

To induce strong scattering, we introduce a $6\,\mu$m-long disordered region containing randomly distributed air holes. The hole radii are uniformly distributed between 20 and 50~nm, with a scatterer density of $20\,\mu\text{m}^{-2}$ and a minimum edge-to-edge separation of 40~nm.

For each waveguide, the input partially coherent wave is described by a density matrix $\rho = \sum_{i=1}^{n} \lambda_i |\psi_i\rangle\langle\psi_i|$, representing a statistical mixture of the $n$ guided eigenmodes $|\psi_i\rangle$ with weights $\lambda_i$. We then apply unitary control to generate any wave with the same power and coherence spectrum as $\rho$: $\rho [U] = U \rho U^{\dagger}, \, U \in U(n)$.  

We feed $\rho[U]$ into the waveguide and analyze energy delivery to two target regions, $A$ and $B$, located deep within the scattering medium. In the $n = 2$ geometry, regions $A$ and $B$ are $0.3\,\mu\text{m} \times 0.3\,\mu$m squares centered at depths of $5\,\mu$m and $5.3\,\mu$m, respectively. In the $n = 3$ geometry, the corresponding regions are $0.48\,\mu\text{m} \times 0.48\,\mu$m squares located at depths of $4.5\,\mu$m and $5\,\mu$m. The simulations are performed using the finite-difference frequency-domain (FDFD) method via the MESTI solver~\cite{MESTI,lin2022_APF}.

We compute the transmission matrices $t_A$ and $t_B$, which map the incident eigenmodes $|\psi_i\rangle$ to the internal fields within regions $A$ and $B$, respectively. These matrices have dimensions $M_A \times n$ and $M_B \times n$, where $M_{A/B}$ denotes the number of spatial grid points in each region. Following Ref.~\cite{shaughnessy2024multiregion}, we define the measurement operators $I_A = t_A^\dagger t_A$ and $I_B = t_B^\dagger t_B$, corresponding to the integrated intensity in each region. We normalize these operators by the average of their eigenvalues, representing the mean intensity under random illumination. The resulting eigenvalues $\lamd{I_A}$ and $\lamd{I_B}$ for both the $n = 2$ and $n = 3$ cases are listed in~Appendix~\ref{sec:numerical_values}.

\begin{figure*}[htbp]
\centering
\includegraphics[width=0.65\textwidth]{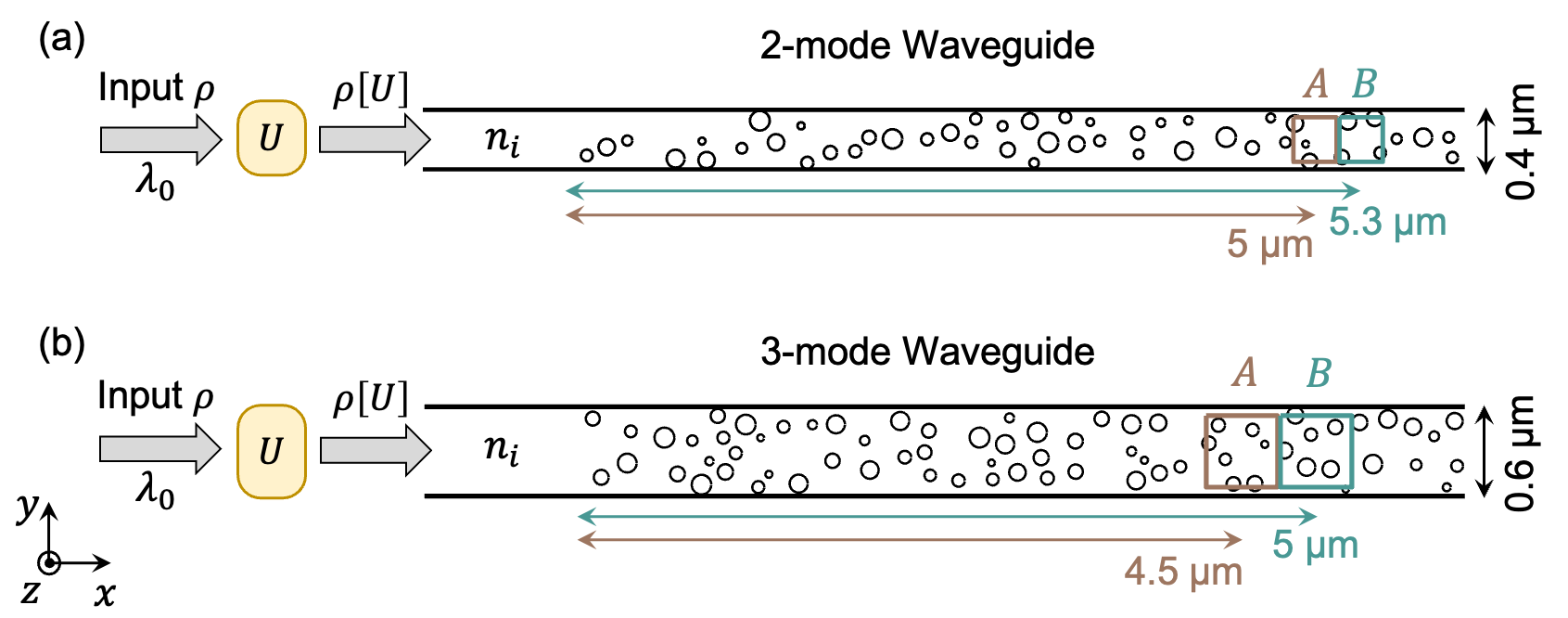}
\caption{Simulation setup for power delivery in diffusive waveguides. TM-polarized light at vacuum wavelength $\lambda_0 = 1.55\,\mu$m is injected from the left into silicon waveguides (refractive index $n_i = 3.48$, embedded in air) containing a $6\,\mu$m-long disordered region of randomly distributed air holes. (a) A $0.4\,\mu$m-wide waveguide supporting $n = 2$ eigenmodes. Target regions $A$ and $B$ (each $0.3\,\mu\text{m} \times 0.3\,\mu$m) are centered at depths of $5\,\mu$m and $5.3\,\mu$m, respectively. (b) A $0.6\,\mu$m-wide waveguide supporting $n = 3$ eigenmodes. Target regions $A$ and $B$ (each $0.48\,\mu\text{m} \times 0.48\,\mu$m) are centered at depths of $4.5\,\mu$m and $5\,\mu$m, respectively.}
\label{fig:Fig4}
\end{figure*}

\section{Data for Figs.~\ref{fig:Fig1} and~\ref{fig:Fig3}}\label{sec:numerical_values}

\subsection{Data for Fig.~\ref{fig:Fig1}}

For $n=2$, the eigenvalues of $I_A$ and $I_B$ are
\begin{equation}
    \lamd{I_A} = (1.47, 0.53), \quad  
    \lamd{I_B} = (1.36, 0.64).    
\end{equation}
The eigenvalues of inputs $\{\rho\}=\{\rho_a, \rho_b, \rho_c\}$ are
\begin{align}
    &\lamd{\rho_a} = (0.81, 0.19),\quad \lamd{\rho_b} = (0.71, 0.29), \quad \lamd{\rho_c} = (0.61, 0.39).
\end{align}
For $n=3$, the eigenvalues of $I_A$ and $I_B$ are
\begin{align}
 \lamd{I_A} = (1.81, 0.61, 0.58), \quad \lamd{I_B} = (1.41, 1.39, 0.19).     
\end{align}
The eigenvalues of inputs $\{\rho\}=\{\rho_d, \rho_e, \rho_f\}$ are
\begin{align}
    \lamd{\rho_d} = (0.80, 0.15, 0.05), \quad
    \lamd{\rho_e} = (0.55, 0.45, 0.00), \quad
    \lamd{\rho_f} = (0.60, 0.30, 0.10).
\end{align}
The ranges of achievable measured values $\{I_A\}$ and $\{I_B\}$ for partially coherent waves are calculated using Eq.~(\ref{eq:k_range}).

\subsection{Data for Fig.~\ref{fig:Fig3}}

The set $\Lambda$ in Fig.~\ref{fig:Fig3}(a) is an elliptical disk in $\Delta_3^\downarrow$ centered at $(0.55, 0.32, 0.13)$, with major axis of length $0.32$ along the $(0.05, 0.04, -0.09)$ direction and minor axis of length $0.12$. Its supremum and infimum are 
\begin{equation}
\bm{\lambda}^{\downarrow}_{\sup} = (0.65, 0.29, 0.06), \quad  \bm{\lambda}^{\downarrow}_{\inf} = (0.45, 0.36, 0.19).
\end{equation}
Its minimum and maximum entropy elements are
\begin{align}
\bm{\lambda}^{\downarrow}_{H_{\min}} = (0.55, 0.38, 0.07), \quad  
\bm{\lambda}^{\downarrow}_{H_{\max}} = (0.50, 0.33, 0.17).
\end{align}

The set $\Lambda$ in Fig.~\ref{fig:Fig3}(d) is a regular hexagon in $\Delta^{\downarrow}_3$ centered at $(0.525, 0.350, 0.125)$, with vertices at 
\begin{align}
&(0.525, 0.275, 0.200), \quad (0.600, 0.275, 0.125), \quad
(0.600, 0.350, 0.050), \\
&(0.525, 0.425, 0.050), \quad (0.450, 0.425, 0.125), \quad (0.450, 0.350, 0.200).  
\end{align}
Its supremum and infimum are
\begin{align}
    \bm{\lambda}^{\downarrow}_{\sup} = (0.600,0.350,0.050),  \quad
    \bm{\lambda}^{\downarrow}_{\inf} = (0.450,0.350,0.200).
\end{align}

\section{Proofs of Eqs.~(\ref{eq:restrict_outer_bounds}) and (\ref{eq:restrict_inner_bounds})}\label{sec:proof_maxmin}

To prove Eqs.~(\ref{eq:restrict_outer_bounds}) and (\ref{eq:restrict_inner_bounds}), we first establish a lemma:

\begin{lemma*}
Let $\Lambda \subseteq \Delta_n^{\downarrow}$ be nonempty and compact. Denote by $\Lambda_{\min}$ and $\Lambda_{\max}$ the sets of minimal and maximal elements of $\Lambda$ in the majorization order. Then for any $\bm{x} \in \Lambda$, there exist $\bm{x}_m \in \Lambda_{\min}$ and $\bm{x}_M \in \Lambda_{\max}$ such that 
\begin{equation}\label{eq:nxmbounds}
\bm{x}_m \prec \bm{x} \prec \bm{x}_M.   
\end{equation}
In particular, $\Lambda_{\min}$ and $\Lambda_{\max}$ are nonempty.
\end{lemma*}

\begin{proof}
We first recall a standard result from majorization theory.
Let $\varphi:\mathbb{R}\to\mathbb{R}$ be convex and define
$\Phi(\bm{z}) \coloneqq \sum_{i=1}^n \varphi(z_i)$.
If $\bm{u},\bm{v}\in\mathbb{R}^n$ satisfy $\bm{u} \prec \bm{v}$, then
\begin{equation}\label{eq:Schur-ineq}
  \Phi(\bm{u}) \le \Phi(\bm{v}).
\end{equation}
Moreover, if $\varphi$ is strictly convex and $\bm{u}$ is not a permutation of $\bm{v}$,
then the inequality in \eqref{eq:Schur-ineq} is strict; see Ref.~\cite{marshall2011}, Ch.~3.
On $\Delta_n^\downarrow$, vectors are arranged in nonincreasing order,
so ``permutation'' reduces to equality.

We apply this result with the strictly convex function $\varphi(t)=t^2$ and define
\begin{equation}\label{eq:defF}
  F(\bm{z}) \coloneqq \sum_{i=1}^n z_i^2, \qquad \bm{z}\in\Delta_n^\downarrow.
\end{equation}
Then for all $\bm{u},\bm{v}\in\Delta_n^\downarrow$,
\begin{equation}\label{eq:Schur-strict}
  \bm{u}\prec \bm{v}, \, \bm{u}\ne \bm{v} \quad\Longrightarrow\quad F(\bm{u}) < F(\bm{v}).
\end{equation}
So $F$ is strictly increasing along the majorization order.

\medskip
\noindent
\textit{Minimal element below a given point.}
Fix $\bm{x}\in\Lambda$ and consider the lower set
\begin{equation}\label{eq:Dx}
  D_x \coloneqq \{\bm{y}\in\Lambda : \bm{y} \prec \bm{x}\}.
\end{equation}
We have $\bm{x}\in D_x$, so $D_x$ is nonempty.
The relation $\bm{y}\prec \bm{x}$ is defined by finitely many linear inequalities
in the coordinates of $\bm{y}$ [see Eq.~(\ref{eq:majorization_1})], hence
$\{\bm{y}\in\Delta_n^\downarrow : \bm{y}\prec \bm{x}\}$ is closed in $\Delta_n^\downarrow$.
Therefore, $D_x = \Lambda \cap \{\bm{y} : \bm{y}\prec \bm{x}\}$ is a closed subset of the
compact set $\Lambda$ and is itself compact.

By continuity of $F$, there exists $\bm{x}_m \in D_x$ such that
\begin{equation}\label{eq:Fn-min}
  F(\bm{x}_m) = \min_{\bm{y}\in D_x} F(\bm{y}).
\end{equation}
If $\bm{z}\in D_x$ and $\bm{z}\prec \bm{x}_m$ with $\bm{z}\ne \bm{x}_m$, then \eqref{eq:Schur-strict}
would imply $F(\bm{z})<F(\bm{x}_m)$, contradicting \eqref{eq:Fn-min}.
Hence, $\bm{x}_m$ is minimal in $D_x$ with respect to $\prec$.
If $\bm{w}\in\Lambda$ satisfied $\bm{w}\prec \bm{x}_m$ with $\bm{w}\ne \bm{x}_m$, then using $\bm{x}_m\prec \bm{x}$ and the
transitivity of $\prec$, we would obtain
$\bm{w}\prec \bm{x}$, i.e., $\bm{w}\in D_x$, contradicting the minimality of $\bm{x}_m$
in $D_x$.
Thus, $\bm{x}_m$ is minimal in $\Lambda$, so $\bm{x}_m\in\Lambda_{\min}$, and by
construction $\bm{x}_m\prec \bm{x}$.

\medskip
\noindent
\textit{Maximal element above a given point.}
Similarly, consider the upper set
\begin{equation}\label{eq:Ux}
  U_x \coloneqq \{\bm{y}\in\Lambda : \bm{x} \prec \bm{y}\}.
\end{equation}
Again $\bm{x}\in U_x$, so $U_x$ is nonempty, and the same closedness argument
shows that $U_x$ is compact.

By continuity of $F$, there exists $\bm{x}_M \in U_x$ such that
\begin{equation}\label{eq:Fm-max}
  F(\bm{x}_M) = \max_{\bm{y}\in U_x} F(\bm{y}).
\end{equation}
If $\bm{z}\in U_x$ and $\bm{x}_M\prec \bm{z}$ with $\bm{x}_M\ne \bm{z}$, then
\eqref{eq:Schur-strict} gives $F(\bm{x}_M)<F(\bm{z})$, contradicting \eqref{eq:Fm-max}.
Thus, $\bm{x}_M$ is maximal in $U_x$.
If $\bm{w}\in\Lambda$ satisfied $\bm{x}_M\prec \bm{w}$ with $\bm{x}_M\ne \bm{w}$, then from $\bm{x}\prec \bm{x}_M$
and transitivity of $\prec$, we would obtain $\bm{x}\prec \bm{w}$, i.e., $\bm{w}\in U_x$,
contradicting the maximality of $\bm{x}_M$ in $U_x$.
Hence, $\bm{x}_M$ is maximal in $\Lambda$, so $\bm{x}_M\in\Lambda_{\max}$, and by
construction $\bm{x}\prec \bm{x}_M$.

\medskip
We have shown that for each $\bm{x}\in\Lambda$ there exist
$\bm{x}_m \in\Lambda_{\min}$ and $\bm{x}_M \in\Lambda_{\max}$ such that
$\bm{x}_m \prec \bm{x} \prec \bm{x}_M$. Finally, taking a minimizer (maximizer) of $F$ on the nonempty compact set $\Lambda$ shows that $\Lambda_{\min}$ and $\Lambda_{\max}$ are nonempty.
\end{proof}

Now we prove Eqs.~(\ref{eq:restrict_outer_bounds}) and (\ref{eq:restrict_inner_bounds}).

\begin{proof}
We first prove Eq.~(\ref{eq:restrict_outer_bounds}). According to the lemma, for any $k' \in \Lambda \backslash \Lambda_{\max}$, there exists $k \in \Lambda_{\max}$ such that $k' \prec k$. Using Eq.~(\ref{eq:majorization_order_preserved}), we obtain 
\begin{equation}
\{o\}_{k'} \subseteq \{o\}_{k}.
\end{equation}
Therefore,
\begin{equation}
\bigcup_{k' \in \Lambda \backslash \Lambda_{\max}} \{o\}_{k'}  \subseteq  \bigcup_{k \in \Lambda_{\max}} \{o\}_k,
\end{equation}
which implies Eq.~(\ref{eq:restrict_outer_bounds}). The proof of Eq.~(\ref{eq:restrict_inner_bounds}) follows by an analogous argument.
\end{proof}
A numerical demonstration using the $n=3$ elliptical disk in Fig.~\ref{fig:Fig3}(a) confirms that the outer and inner bounds are attained by maximal and minimal elements [Fig.~\ref{fig:Fig5}].

\begin{figure}[t]
\centering
\includegraphics[width=0.5\textwidth]{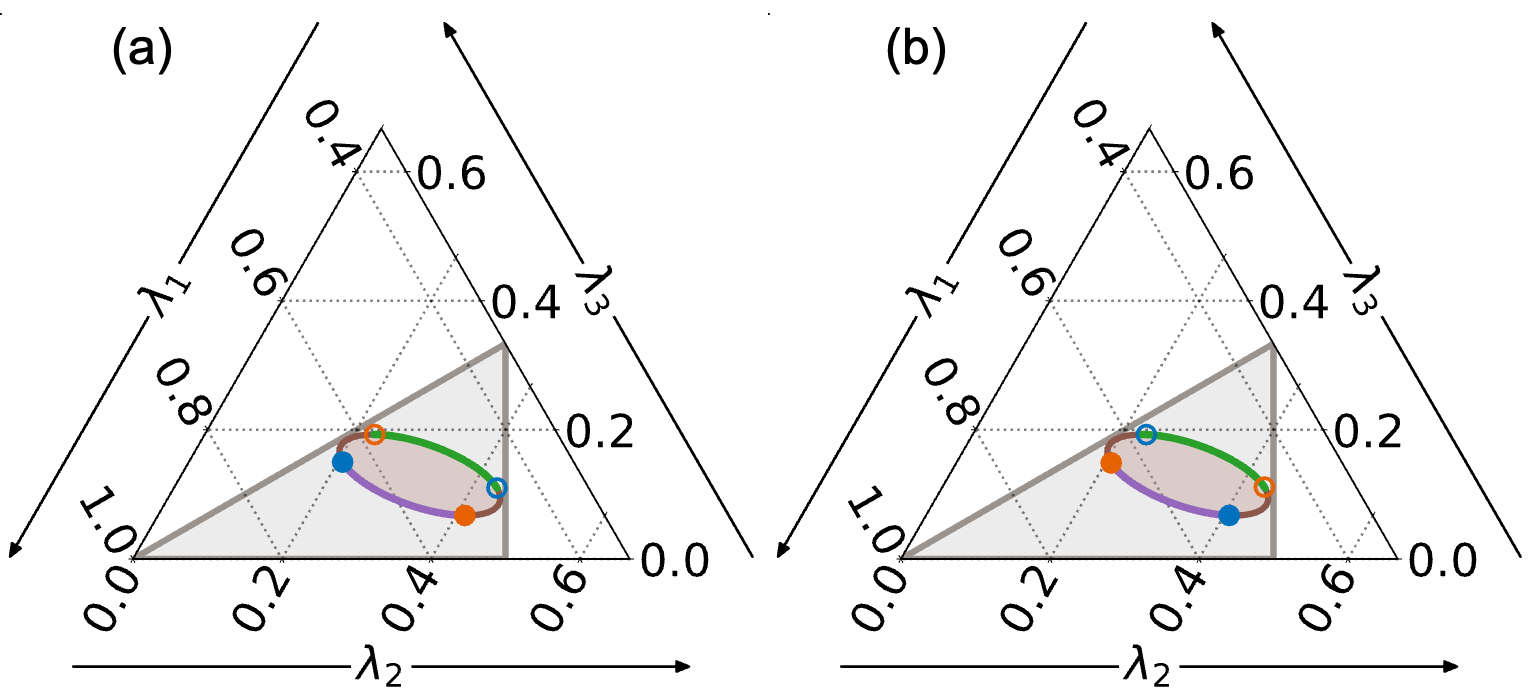}
\caption{Coherence spectra that attain the outer bounds (filled circles) and inner bounds (empty circles) for (a) $I_A$ and (b) $I_B$ measurements for the $n=3$ elliptical disk set in Fig.~\ref{fig:Fig3}(a). These spectra are elements of $\Lambda_{\max}$ (purple) and $\Lambda_{\min}$ (green), confirming that maximal and minimal elements attain the exact outer and inner bounds.}
\label{fig:Fig5}
\end{figure}

\section{Additional numerical demonstrations of supremum and infimum bounds}\label{sec:additional_ND}

In this section, we provide additional numerical demonstrations of the universality and optimality of the supremum and infimum bounds, complementing Figs.~\ref{fig:Fig1} and~\ref{fig:Fig3}.

\subsection{Finite set}

\begin{figure}[htbp]
\centering
\includegraphics[width=0.6\textwidth]{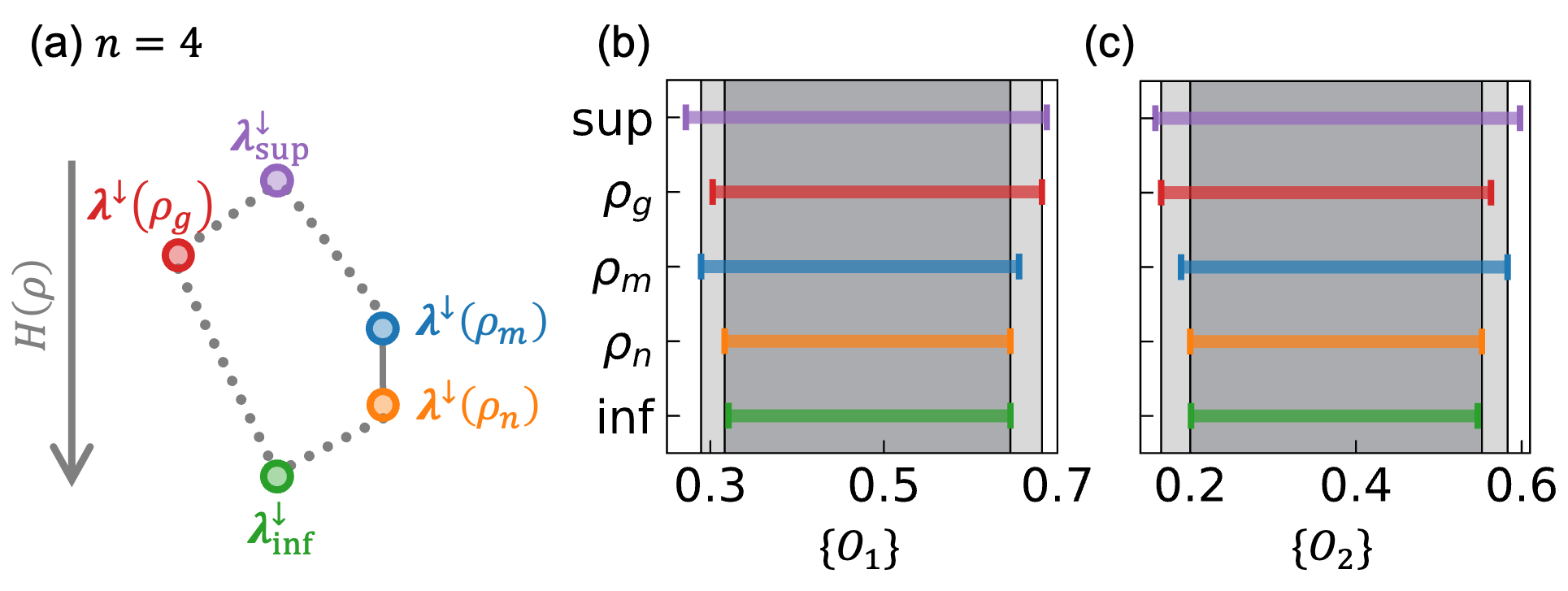}
\caption{Numerical demonstration for an $n=4$ finite set $\Lambda$. (a) Hasse diagram showing the majorization relations among coherence spectra $\lamd{\rho_g}$, $\lamd{\rho_m}$, and $\lamd{\rho_n}$. An edge indicates a strict majorization relation, and a higher vertical position corresponds to a lower entropy $H(\rho)$. The supremum $\bm{\lambda}^{\downarrow}_{\mathrm{sup}}$ and infimum $\bm{\lambda}^{\downarrow}_{\mathrm{inf}}$ of the set are also shown. (b,c) Achievable measured values $\{O_1\}$ and $\{O_2\}$ for two measurement operators. The supremum and infimum bounds (purple and green intervals) satisfy $\{o\}_{\mathrm{inf}} \subseteq \{o\}_k \subseteq \{o\}_{\mathrm{sup}}$ for $k \in \{g, m, n\}$. The exact outer and inner bounds are indicated by light and dark shades.}
\label{fig:Fig6}
\end{figure}

Figure~\ref{fig:Fig6} shows a finite set $\Lambda$ with $n=4$ comprising partially coherent waves with coherence spectra
\begin{equation}
    \lamd{\rho_g} = (0.55, 0.35, 0.08, 0.02), \quad
    \lamd{\rho_m} = (0.62, 0.17, 0.15, 0.06), \quad
    \lamd{\rho_n} = (0.56, 0.23, 0.14, 0.07).
\end{equation}
The supremum and infimum of $\Lambda$ are
\begin{equation}
    \bm{\lambda}^{\downarrow}_{\sup} = (0.62, 0.28, 0.08, 0.02), \quad \bm{\lambda}^{\downarrow}_{\inf} = (0.55, 0.24, 0.14, 0.07).
\end{equation}
We obtain the ranges of achievable measured values from Eq.~(\ref{eq:k_range}) for two measurement operators with eigenvalues
\begin{equation}
    \lamd{O_1} = (0.74, 0.66, 0.53, 0.09),\quad
    \lamd{O_2} = (0.80, 0.30, 0.20, 0.10).
\end{equation}

The supremum and infimum bounds (purple and green intervals) satisfy $\{ o \}_{\inf} \subseteq \{ o \}_k \subseteq \{ o \}_{\sup}$ for $k \in \{ g,m,n \}$ in both measurements, closely matching the exact outer and inner bounds (light and dark shades):
\begin{align}
    &\{O_1\}_{\sup} = [0.27, 0.69], \quad &\{O_1\}_{\mathrm{outer}} = [0.29, 0.68]; \qquad
    &\{O_1\}_{\inf} = [0.32, 0.65], \quad &\{O_1\}_{\mathrm{inner}} = [0.32, 0.65].\\
    &\{O_2\}_{\sup} = [0.16, 0.60], \quad &\{O_2\}_{\mathrm{outer}} = [0.17, 0.58]; \qquad
    &\{O_2\}_{\inf} = [0.20, 0.55], \quad &\{O_2\}_{\mathrm{inner}} = [0.20, 0.55].
\end{align}

\subsection{Closed infinite set}

\begin{figure}[htbp]
\centering
\includegraphics[width=0.6\textwidth]{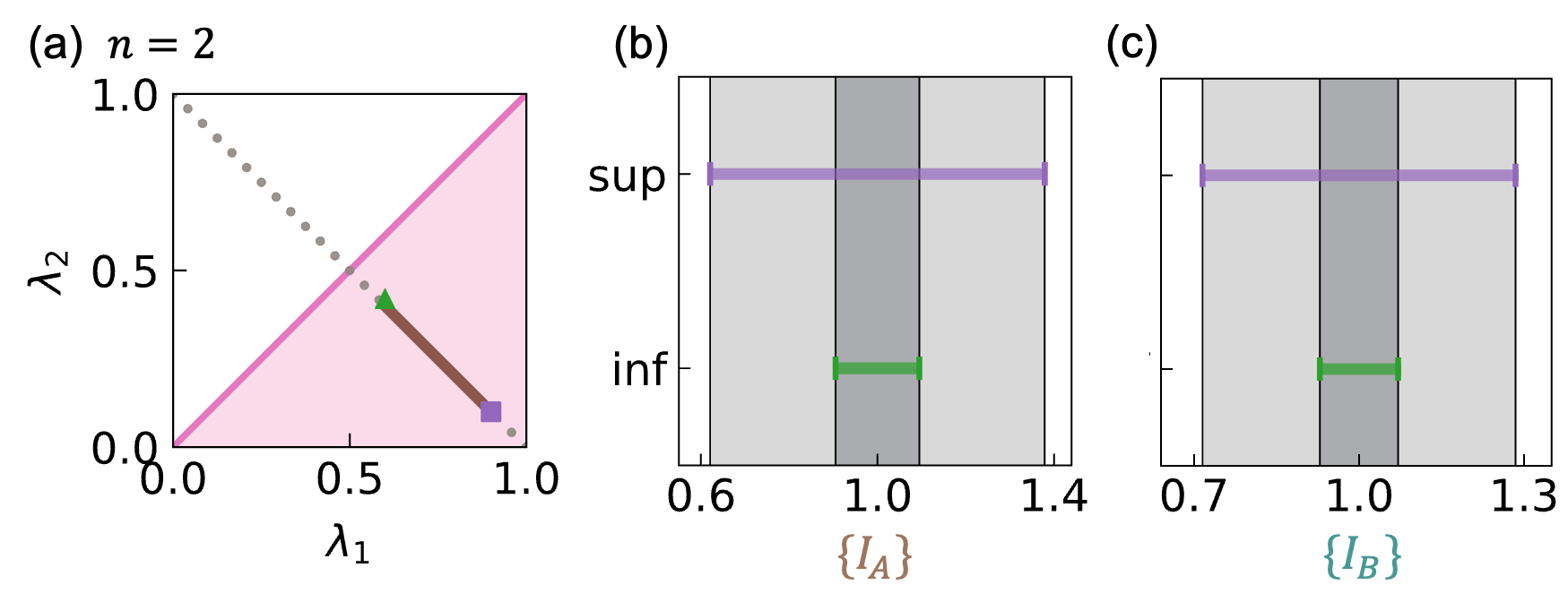}
\caption{Numerical example for a closed infinite set $\Lambda$ with $n=2$. (a) The set $\Lambda$ is a line segment with endpoints $(0.90,0.10)$ and $(0.60,0.40)$. The supremum $\bm{\lambda}^{\downarrow}_{\mathrm{sup}}$ (purple square) and infimum $\bm{\lambda}^{\downarrow}_{\mathrm{inf}}$ (green triangle) are its endpoints. (b,c) Achievable measured values $\{I_A\}$ and $\{I_B\}$. The supremum and infimum bounds (purple and green intervals) coincide with the exact outer and inner bounds (light and dark shades), respectively.}
\label{fig:Fig7}
\end{figure}
Figure~\ref{fig:Fig7} shows a closed infinite set with $n=2$. The set $\Lambda$ is a line segment in $\Delta_2^\downarrow$ with endpoints $(0.90,0.10)$ and $(0.60,0.40)$. The supremum and infimum are its endpoints:
\begin{equation}
\bm{\lambda}^{\downarrow}_{\sup} = (0.90,0.10), \quad \bm{\lambda}^{\downarrow}_{\inf} = (0.60,0.40).  
\end{equation}
The outer and inner bounds (light and dark shades) are attained by $\bm{\lambda}^{\downarrow}_{\sup}$ and $\bm{\lambda}^{\downarrow}_{\inf}$ (purple and green intervals) for both $I_A$ and $I_B$ measurements.

\begin{figure}[tb]
\centering
\includegraphics[width=0.22\textwidth]{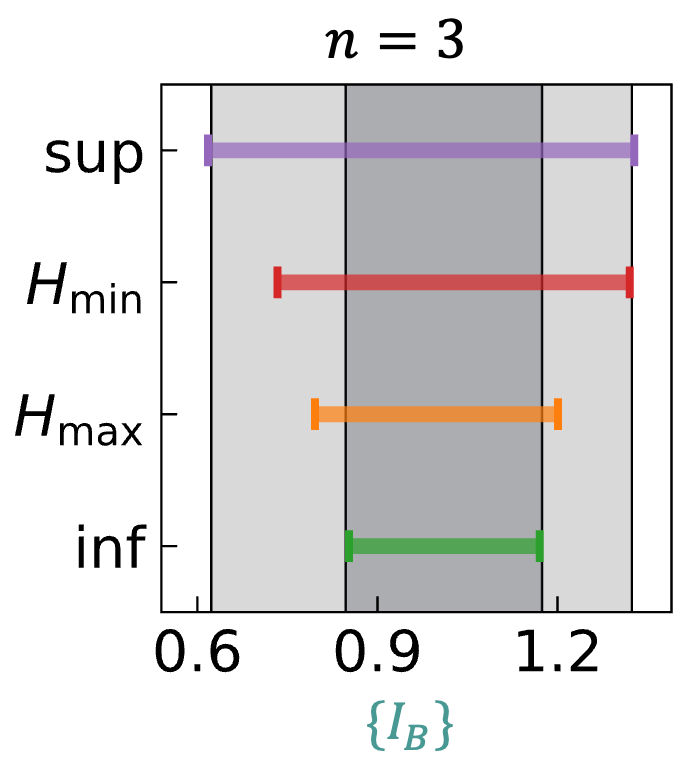}
\caption{Achievable measured values $\{I_B\}$ for the $n=3$ elliptic disk set from Fig.~\ref{fig:Fig3}(a), complementing the $\{I_A\}$ results in Fig.~\ref{fig:Fig3}(b). The supremum and infimum bounds (purple and green intervals) closely match the exact outer and inner bounds (light  and dark shades), while the entropy minimum ($H_{\min}$) and maximum ($H_{\max}$) elements fail to bound the measurement.}
\label{fig:Fig8}
\end{figure}

Figure~\ref{fig:Fig3}(b) shows the range of achievable measured values $\{I_A\}$ for the $n=3$ set in Fig.~\ref{fig:Fig3}(a). Figure~\ref{fig:Fig8} shows the corresponding results for $\{I_B\}$. The supremum and infimum bounds (purple and green intervals) also match the exact outer and inner bounds (light and dark shades):
\begin{align}
    &\{I_B\}_{\sup} = [0.62, 1.33], \quad &\{I_B\}_{\mathrm{outer}} = [0.62, 1.32]; \qquad
    &\{I_B\}_{\inf} = [0.85, 1.17], \quad &\{I_B\}_{\mathrm{inner}} = [0.85, 1.17].
\end{align}

Figure~\ref{fig:Fig9} provides an additional example of a closed infinite set with $n=3$. The set $\Lambda$ is formed by the union of the elliptical disk from Fig.~\ref{fig:Fig3}(a) and an external point $(0.76, 0.15, 0.09)$.
The supremum and infimum are
\begin{equation}
    \bm{\lambda}^{\downarrow}_{\sup} = (0.76, 0.19, 0.05), \quad
    \bm{\lambda}^{\downarrow}_{\inf} = (0.46, 0.35, 0.19).
\end{equation}
Both lie outside $\Lambda$. They yield bounds (purple and green intervals) that tightly match the exact outer and inner bounds (light and dark shades) for both $I_A$ and $I_B$ measurements:
\begin{align}
    &\{I_A\}_{\sup} = [0.65, 1.52], \quad &\{I_A\}_{\mathrm{outer}} = [0.67, 1.52]; \qquad
    &\{I_A\}_{\inf} = [0.83, 1.15], \quad &\{I_A\}_{\mathrm{inner}} = [0.82, 1.16].\\
    &\{I_B\}_{\sup} = [0.48, 1.35], \quad &\{I_B\}_{\mathrm{outer}} = [0.48, 1.32]; \qquad
    &\{I_B\}_{\inf} = [0.85, 1.17], \quad &\{I_B\}_{\mathrm{inner}} = [0.85, 1.17].
\end{align}

In contrast, the entropy minimum and maximum elements
\begin{equation}
    \bm{\lambda}^{\downarrow}_{H_{\min}} = (0.76, 0.15, 0.09), \quad
    \bm{\lambda}^{\downarrow}_{H_{\max}} = (0.50, 0.33, 0.17).
\end{equation}
are elements of $\Lambda$ but fail to bound either measurement accurately.

\begin{figure}[htbp]
\centering
\includegraphics[width=0.67\textwidth]{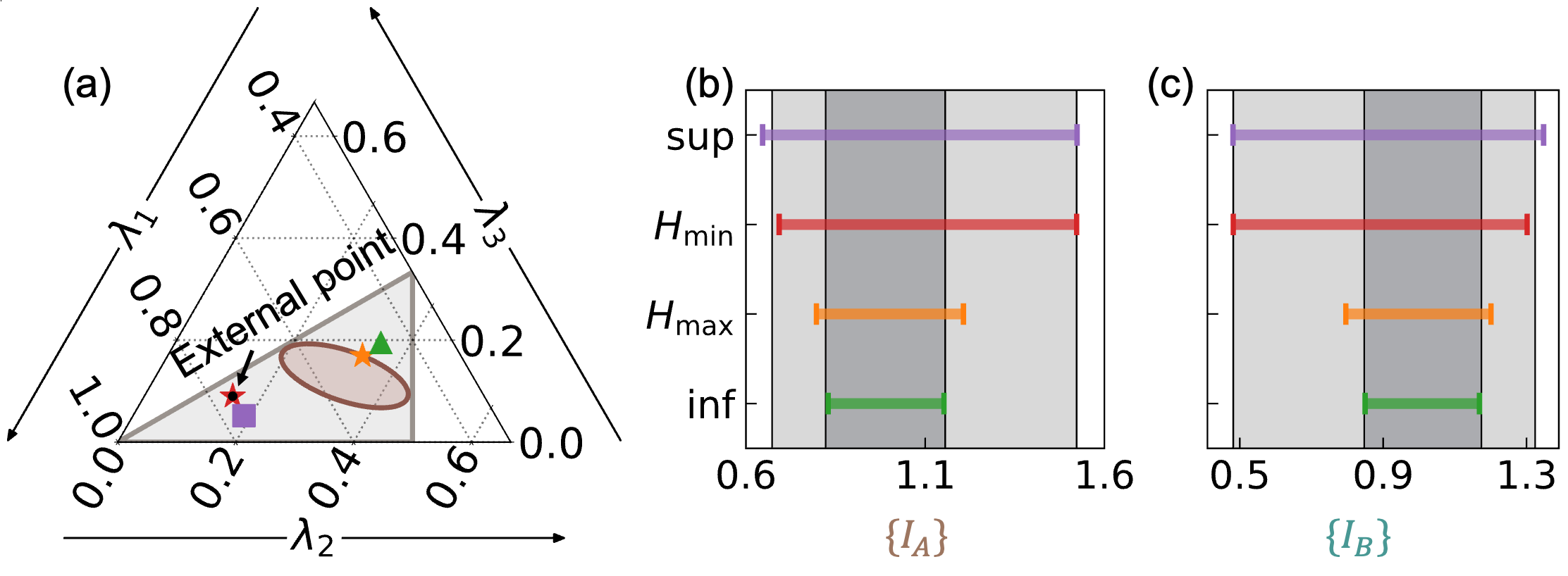}
\caption{Additional numerical example for a closed infinite set $\Lambda$ with $n=3$. (a) The set $\Lambda$ is the union of the elliptic disk from Fig.~\ref{fig:Fig3}(a) and an external point $(0.76, 0.15, 0.09)$. Both the supremum $\bm{\lambda}^{\downarrow}_{\mathrm{sup}}$ (purple square) and the infimum $\bm{\lambda}^{\downarrow}_{\mathrm{inf}}$ (green triangle) lie outside $\Lambda$. (b,c) Achievable measured values $\{I_A\}$ and $\{I_B\}$. The supremum and infimum bounds (purple and green intervals) tightly match the exact outer and inner bounds (light and dark shades) for both measurements. In contrast, the entropy minimum ($H_{\min}$) and maximum ($H_{\max}$) elements of $\Lambda$ fail to bound either measurement accurately.}
\label{fig:Fig9}
\end{figure}

\section{Construction of inscribing and circumscribing polygons}
\label{sec:inscribing_circumscribing}

\begin{figure}[b]
\centering
\includegraphics[width=0.27\textwidth]{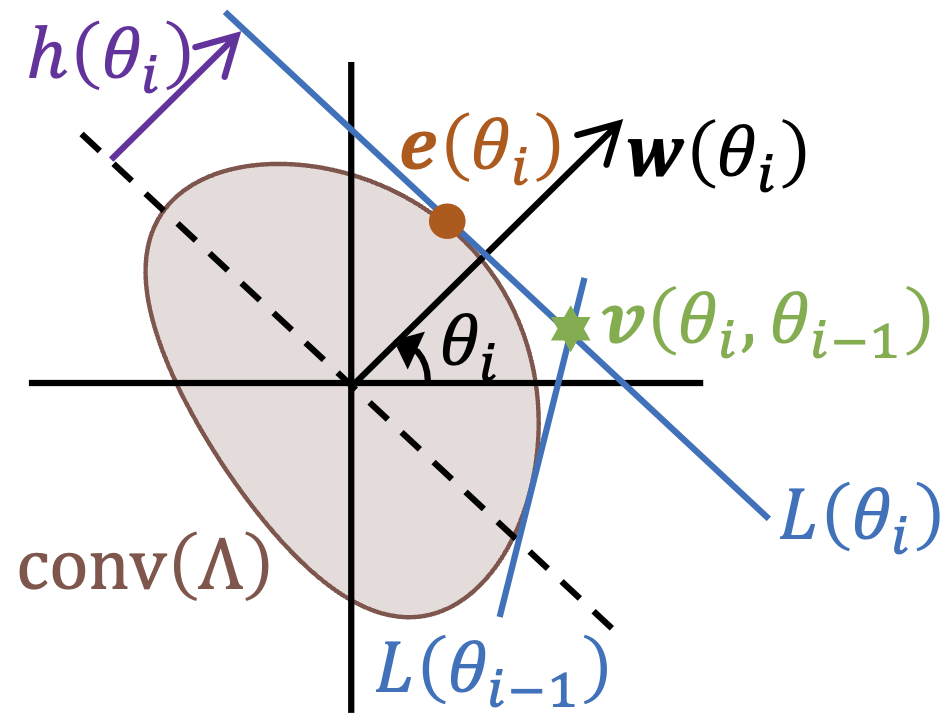}
\caption{Geometric construction of inscribing and circumscribing polygons for a 2D convex set $\operatorname{conv}(\Lambda)$. For a given direction $\bm{w}(\theta_i)$, the supporting line $L(\theta_i)$ (blue) touches the boundary at the tangent point $\bm{e}(\theta_i)$, with $h(\theta_i)$ denoting the support function value. The collection of all such boundary points $\{\bm{e}(\theta_i)\}$ defines the inscribing polygon. The intersection of adjacent supporting lines $L(\theta_{i-1})$ and $L(\theta_i)$ defines a vertex $\bm{v}(\theta_i, \theta_{i-1})$ of the circumscribing polygon. The collection of all such vertices $\bm{v}(\theta_i, \theta_{i-1})$ defines the circumscribing polygon.}
\label{fig:Fig10}
\end{figure}

In this section, we present algorithms to construct inscribing and circumscribing polygons for the convex hull of a set $\Lambda \subseteq \Delta^{\downarrow}_3$. We focus on the cases where the convex hull $\operatorname{conv}(\Lambda)$ is a 2D region embedded in $\Delta^{\downarrow}_3$.

\textit{Inscribing polygon.}---We construct the inscribing polygon by identifying boundary points of $\operatorname{conv}(\Lambda)$ that are tangent to a series of supporting lines. For each angle $\theta_i = i(2\pi/N)$ with $i = 1, \ldots, N$, we define a direction vector $\bm{w}(\theta_i) = (\cos \theta_i, \sin \theta_i)$. The supporting line orthogonal to $\bm{w}(\theta_i)$ is given by
\begin{equation}
L(\theta_i) = \{\bm{x} \in \operatorname{conv}(\Lambda) \mid \bm{x} \cdot \bm{w}(\theta_i) = h(\theta_i)\},
\label{eq:support_line}
\end{equation}
where the support function $h(\theta_i)$ is defined as
\begin{equation}
h(\theta_i) = \max_{\bm{x} \in \operatorname{conv}(\Lambda)} \bm{x} \cdot \bm{w}(\theta_i).
\label{eq:support_function}
\end{equation}
As shown in Fig.~\ref{fig:Fig10}, the supporting line $L(\theta_i)$ (blue line) keeps the entire $\operatorname{conv}(\Lambda)$ on one side while touching the boundary at point $\bm{e}(\theta_i)$ without intersecting the interior. $h(\theta_i)$ represents the maximum projection of any point in $\operatorname{conv}(\Lambda)$ along the direction $\bm{w}(\theta_i)$.

To compute $h(\theta_i)$ and $\bm{e}(\theta_i)$, we parameterize the boundary of $\operatorname{conv}(\Lambda)$ as $\bm{p}(t) = (x(t), y(t))$ with a scalar parameter $t \in [t_0, t_1]$.
Equation~\eqref{eq:support_function} then becomes
\begin{equation}
h(\theta_i) = \max_{t \in [t_0, t_1]} \bm{p}(t) \cdot \bm{w}(\theta_i).
\label{eq:support_function_t}
\end{equation}
The parameter $t^*$ that maximizes this projection satisfies the gradient condition
\begin{equation}
\frac{d(\bm{p} \cdot \bm{w})}{dt} = \frac{dx}{dt} \cos \theta_i + \frac{dy}{dt} \sin \theta_i = 0.
\label{eq:derivative_G}
\end{equation}
The support function is then $h(\theta_i) = \bm{p}(t^*) \cdot \bm{w}(\theta_i)$, and the tangent boundary point is $\bm{e}(\theta_i) = \bm{p}(t^*)$. The collection of all such boundary points $\{\bm{e}(\theta_i)\}$ over the full $2\pi$ angular range forms the inscribing polygon.

\textit{Circumscribing polygon.}---We construct the circumscribing polygon by finding the intersections of adjacent supporting lines. As illustrated in Fig.~\ref{fig:Fig10}, each vertex $\bm{v}(\theta_{i}, \theta_{i-1})$ (green star) of the circumscribing polygon is given by $\bm{v}(\theta_i, \theta_{i-1}) = L(\theta_{i}) \cap L(\theta_{i-1})$. These vertices can be computed using~\cite{prince1990}
\begin{equation}
\bm{v}(\theta_i, \theta_{i-1}) = \frac{1}{\sin(2\pi/N)} 
\begin{bmatrix}
\sin \theta_i & -\sin \theta_{i-1} \\
-\cos \theta_i & \cos \theta_{i-1}
\end{bmatrix}
\begin{bmatrix}
h(\theta_{i-1}) \\
h(\theta_i)
\end{bmatrix},
\label{eq:circumscribing_vertex}
\end{equation}
where $\theta_i = i(2\pi/N)$ for $i = 1, \ldots, N$ and $\theta_0 = \theta_N$.

\textit{Implementation.}---Algorithms~\ref{alg:inscrib_polygon} and~\ref{alg:circumscrib_polygon} below provide the computational procedures for obtaining the vertices of inscribing and circumscribing polygons. We apply these algorithms to several representative geometric shapes, as shown in Fig.~\ref{fig:Fig11}. For shapes with analytically parameterizable boundaries [Fig.~\ref{fig:Fig11}(a)], we directly apply the algorithms using a single function $\bm{p}(t)$. For shapes formed by unions or intersections [Figs.~\ref{fig:Fig11}(b,c)], we parameterize each boundary segment separately, compute vertices for each piece, and combine them to form the complete polygon. For complex shapes where an analytical boundary parameterization is challenging, one can solve a constrained convex optimization problem~\cite{boyd2004} for each sampled angle following~\eqref{eq:support_function}.

\begin{algorithm}[H]
\caption{Inscribing Polygon Vertices and Support Function}
\label{alg:inscrib_polygon}
\textbf{input:} $\theta \in \mathbb{R}^N$, function handle $\bm{p}(t)$, function handle $d\bm{p}/dt$ \\
\textbf{output:} vertex coordinates of inscribing polygon $\{\bm{e}(\theta)\}$, support function values $\{h(\theta)\}$.
\begin{algorithmic}
\Procedure{INSCRIBE}{$\theta$, $\bm{p}(t)$, $d\bm{p}/dt$}
\State $\{\bm{e}(\theta)\} \leftarrow \{0\}$ 
\State $\{h(\theta)\} \leftarrow \{0\}$ 
\For{$i = 1$ \dots $N$}
\State $\bm{w} \leftarrow (\cos\theta_i, \sin\theta_i)$
\State $t^* \leftarrow $ solve Eq.~(\ref{eq:derivative_G}) \Comment{Solves for the parameter $t^*$ of maximum projection} 
\State $\bm{e}(\theta_i) \leftarrow \bm{p}(t^*)$ \Comment{Gets the boundary point coordinates}
\State $\{\bm{e}(\theta)\} \leftarrow \text{append}(\{\bm{e}(\theta)\}, \bm{e}(\theta_i))$
\State $h(\theta_i) \leftarrow \bm{e}(\theta_i) \cdot \bm{w}$ \Comment{Calculates the support function value}
\State $\{h(\theta)\} \leftarrow \text{append}(\{h(\theta)\}, h(\theta_i))$
\EndFor
\State \textbf{return} $\{\bm{e}(\theta)\}$, $\{h(\theta)\}$
\EndProcedure
\end{algorithmic}
\end{algorithm}

\begin{algorithm}[H]
\caption{Circumscribing Polygon Vertices}
\label{alg:circumscrib_polygon}
\textbf{input:} $\theta \in \mathbb{R}^N$, $\{h(\theta)\}$ \\
\textbf{output:} vertex coordinates of circumscribing polygon $\{\bm{v}(\theta)\}$.
\begin{algorithmic}
\Procedure{CIRCUMSCRIBE}{$\theta$, $\{h(\theta)\}$}
\State $\{\bm{v}(\theta)\} \leftarrow \{0\}$ 
\For{$i = 1$ \dots $N$}
\State $\bm{v}(\theta_{i},\theta_{i-1}) \leftarrow$ Compute using Eq.~(\ref{eq:circumscribing_vertex}) \Comment{Gets the intersection of adjacent supporting lines}
\State $\{\bm{v}(\theta)\} \leftarrow \text{append}(\{\bm{v}(\theta)\}, \bm{v}(\theta_{i},\theta_{i-1}))$
\EndFor
\State \textbf{return} $\{\bm{v}(\theta)\}$
\EndProcedure
\end{algorithmic}
\end{algorithm}

\begin{figure}[h]
\centering
\includegraphics[width=1\textwidth]{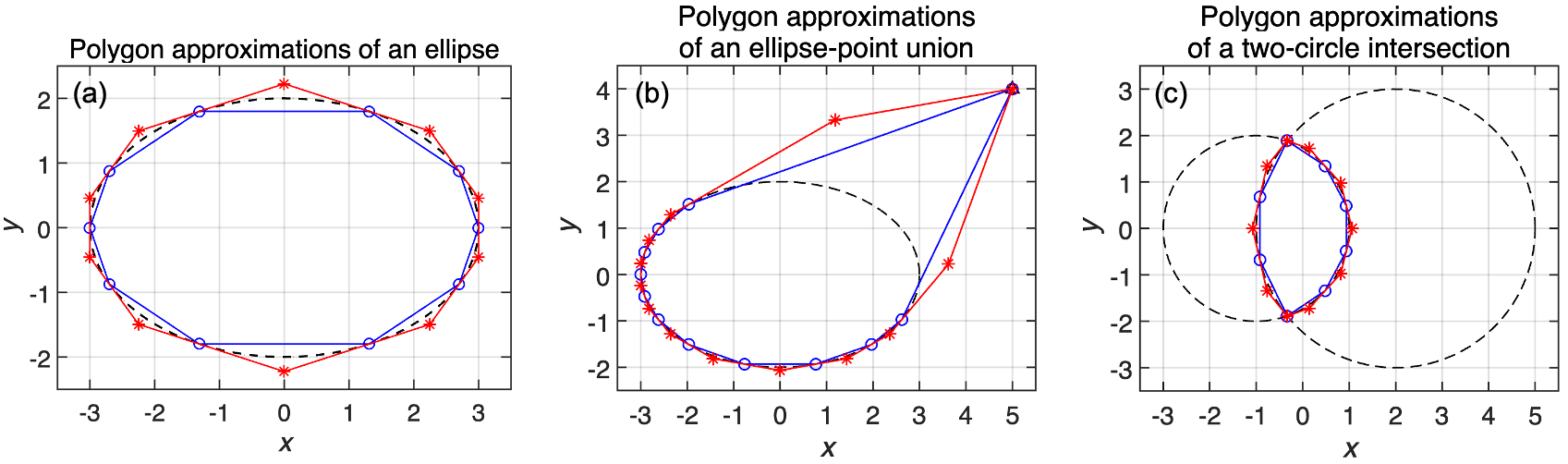}
\caption{Inscribing (blue) and circumscribing (red) polygons for representative geometric shapes (black), constructed using Algorithms~\ref{alg:inscrib_polygon} and~\ref{alg:circumscrib_polygon}. (a) An ellipse with an analytically parameterizable boundary. (b) A union of an ellipse and an external point. (c) An intersection of two circles. }
\label{fig:Fig11}
\end{figure}

\section{Convergence of supremum and infimum computations}\label{sec:convergence}

In this section, we analyze the convergence of supremum and infimum computations using the algorithm in Sec.~\ref{sec:algorithm}.

\begin{figure}[htbp]
\centering
\includegraphics[width=0.65\textwidth]{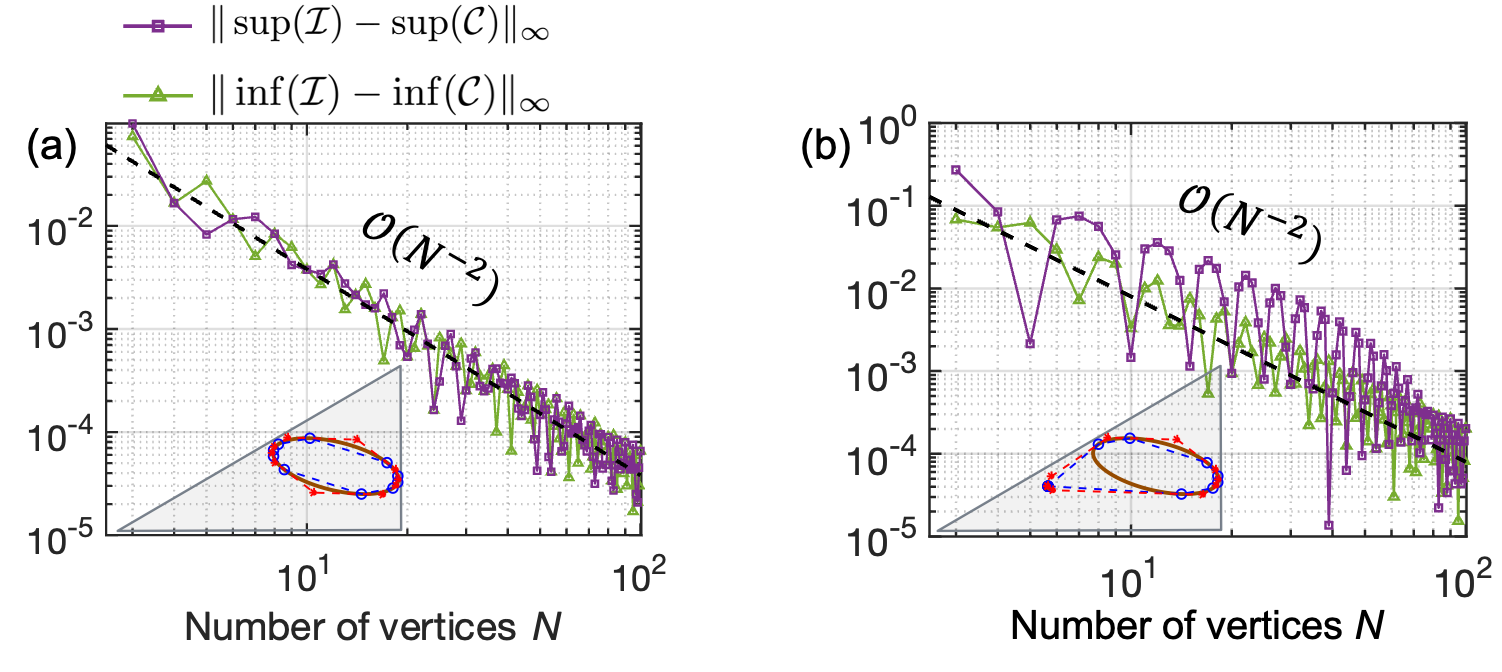}
\caption{Convergence of the supremum and infimum computations as a function of the number of polygon vertices $N$. The plots show $\left\lVert \sup(\mathcal{I}) - \sup(\mathcal{C}) \right\rVert_\infty$ and $\left\lVert \inf(\mathcal{I}) - \inf(\mathcal{C}) \right\rVert_\infty$, where $\mathcal{I}$ and $\mathcal{C}$ denote the inscribing and circumscribing polygons, respectively. (a) The elliptical disk from Fig.~\ref{fig:Fig3}(a). (b) The union shape from Fig.~\ref{fig:Fig9}(a). Black dashed lines indicate $\mathcal{O}(N^{-2})$ convergence. Insets: inscribing (blue) and circumscribing (red) polygons with $N = 10$ vertices approximating the original set $\Lambda$ (brown) within $\Delta^{\downarrow}_3$ (gray).}
\label{fig:Fig12}
\end{figure}

We approximate the convex hull of $\Lambda$ using inscribing ($\mathcal{I}$) and circumscribing ($\mathcal{C}$) polygons with varying numbers of vertices $N$ from 3 to 100. For each $N$, we compute the supremum and infimum of both $\mathcal{I}$ and $\mathcal{C}$, then evaluate their differences $\left\lVert \sup(\mathcal{I}) - \sup(\mathcal{C}) \right\rVert_\infty
$ and $\left\lVert \inf(\mathcal{I}) - \inf(\mathcal{C}) \right\rVert_\infty
$. As shown in Fig.~\ref{fig:Fig12}, increasing the number of polygon vertices $N$ systematically reduces these differences, leading to convergence toward the exact supremum and infimum of the set $\Lambda$. The convergence rate follows a trend of $\mathcal{O}(N^{-2})$, as indicated by the black dashed fitting curves.

\section{Proof of geometric locations of  supremum and infimum}\label{sec:proof_location}
In this section, we prove two results on the locations of $\bm{\lambda}^{\downarrow}_{\sup}$ and $\bm{\lambda}^{\downarrow}_{\inf}$ for a compact set $\Lambda$ with nonempty interior.

\begin{proposition}
\label{thm:sup-inf-not-interior}
Let $\Lambda \subseteq \Delta_n^\downarrow$ ($n\ge 3$) be compact with nonempty interior
with respect to the affine hyperplane
\begin{equation}\label{eq:H-def}
  H \coloneqq \left\{\bm{x}\in\mathbb{R}^n : \sum_{i=1}^n x_i = 1\right\}.
\end{equation}
Then the supremum $\bm{\lambda}^{\downarrow}_{\sup}$ and infimum $\bm{\lambda}^{\downarrow}_{\inf}$ of $\Lambda$ in the majorization order cannot lie in the interior of $\Lambda$.
\end{proposition}

\begin{proof}
We work in the relative topology on $H$. A point $\bm{x}\in\Lambda$ is an
interior point of $\Lambda$ if there exists $\varepsilon>0$ such that
\begin{equation}\label{eq:ball-in-Lambda}
  B_H(\bm{x},\varepsilon) \coloneqq \{\bm{z}\in H : \|\bm{z}-\bm{x}\|<\varepsilon\} \subseteq \Lambda.
\end{equation}
If $\Lambda \subseteq \Delta_n^\downarrow$ and \eqref{eq:ball-in-Lambda} holds,
then $B_H(\bm{x},\varepsilon) \subseteq \Delta_n^\downarrow$, so $\bm{x}$ is
also an interior point of $\Delta_n^\downarrow$ in $H$.

\medskip\noindent
\textbf{Preliminaries.}
Consider the function
\begin{equation}\label{eq:F-def}
  F(\bm{z}) \coloneqq \sum_{i=1}^n z_i^2,\qquad \bm{z}\in\Delta_n^\downarrow.
\end{equation}
The function $F$ is strictly Schur-convex [Eq.~(\ref{eq:Schur-strict})]. See~Appendix~\ref{sec:proof_maxmin} for more details about Schur-convex functions.

The set $\Delta_n^\downarrow$ is a polytope defined by
\begin{equation}\label{eq:Delta-ineq}
  \Delta_n^\downarrow
  = \left\{\bm{x}\in H : x_1 \ge x_2 \ge \cdots \ge x_n \ge 0\right\}.
\end{equation}
Its relative interior in $H$ is defined by strict inequalities:
\begin{equation}\label{eq:Delta-int}
  \operatorname{ri}(\Delta_n^\downarrow)
  = \left\{\bm{x}\in H : x_1 > x_2 > \cdots > x_n > 0\right\}.
\end{equation}
The point
\begin{equation}\label{eq:uniform}
  \bm{u} \coloneqq \Bigl(\frac{1}{n},\dots,\frac{1}{n}\Bigr)
\end{equation}
belongs to $\Delta_n^\downarrow$ but fails to satisfy the strict inequalities
in \eqref{eq:Delta-int}. Hence
\begin{equation}\label{eq:uniform-boundary}
  \bm{u} \text{ lies on the boundary of } \Delta_n^\downarrow \text{ in } H.
\end{equation}

\medskip\noindent
\textbf{Supremum.}
Since $(\Delta^{\downarrow}_n, \prec)$ forms a complete lattice, both $\bm{\lambda}^{\downarrow}_{\sup}$ and $\bm{\lambda}^{\downarrow}_{\inf}$ of $\Lambda$ exist in $\Delta^{\downarrow}_n$. Assume for contradiction that $\bm{x} \coloneqq \bm{\lambda}^{\downarrow}_{\sup}$ lies in
the interior of $\Lambda$ in $H$. Then $\bm{x}$ is an upper bound of $\Lambda$:
\begin{equation}\label{eq:upper-bound}
  \bm{z} \prec \bm{x} \quad\text{for all } \bm{z}\in\Lambda.
\end{equation}
If $\bm{z}\in\Lambda$ and $\bm{z}\neq \bm{x}$, then \eqref{eq:upper-bound} and
\eqref{eq:Schur-strict} yield
\begin{equation}\label{eq:Fx-max}
  F(\bm{z}) < F(\bm{x}) \quad\text{for all } \bm{z}\in\Lambda,\; \bm{z}\neq \bm{x}.
\end{equation}
Thus $\bm{x}$ is the unique global maximizer of $F$ on $\Lambda$.

By the interior assumption, \eqref{eq:ball-in-Lambda} holds for some
$\varepsilon>0$, making $\bm{x}$ an interior local maximizer of $F$ on
$H$. The standard Lagrange multiplier condition for an interior
extremum of $F$ under the constraint
\begin{equation}\label{eq:constraint}
  \sum_{i=1}^n z_i = 1
\end{equation}
implies the existence of $\mu\in\mathbb{R}$ such that
\begin{equation}\label{eq:lagrange}
  \nabla F(\bm{x}) = \mu (1,\dots,1).
\end{equation}
Since
\begin{equation}\label{eq:gradF}
  \nabla F(\bm{z}) = 2\bm{z} \quad\text{for all } \bm{z}\in\mathbb{R}^n,
\end{equation}
\eqref{eq:lagrange} yields
\begin{equation}\label{eq:coords-eq}
  2x_i = \mu \quad\text{for } i=1,\dots,n.
\end{equation}
All coordinates of $\bm{x}$ are therefore equal. Applying \eqref{eq:constraint} at $\bm{z}=\bm{x}$, we
obtain
\begin{equation}\label{eq:x-uniform}
  \bm{x} = \Bigl(\frac{1}{n},\dots,\frac{1}{n}\Bigr) = \bm{u}.
\end{equation}

By \eqref{eq:uniform-boundary}, $\bm{x}$ lies on the boundary of
$\Delta_n^\downarrow$ in $H$ and thus cannot be an interior point of
$\Delta_n^\downarrow$ in $H$. However,
\[
B_H(\bm{x},\varepsilon) \subseteq \Lambda \subseteq \Delta_n^\downarrow
\]
by \eqref{eq:ball-in-Lambda}, thus $\bm{x}$ is an interior point of
$\Delta_n^\downarrow$. This contradiction establishes that
$\bm{\lambda}^{\downarrow}_{\sup}$ cannot lie in the interior of $\Lambda$.

\medskip\noindent
\textbf{Infimum.}
The argument for the infimum is analogous. Assume for contradiction that $\bm{y} \coloneqq \bm{\lambda}^{\downarrow}_{\inf}$
lies in the interior of $\Lambda$ in $H$. Then $\bm{y}$ is a lower
bound of $\Lambda$:
\begin{equation}\label{eq:lower-bound}
  \bm{y} \prec \bm{z} \quad\text{for all } \bm{z}\in\Lambda.
\end{equation}
For $\bm{z}\in\Lambda$ with $\bm{z}\neq \bm{y}$, \eqref{eq:lower-bound} and
\eqref{eq:Schur-strict} yield
\begin{equation}\label{eq:Fy-min}
  F(\bm{y}) < F(\bm{z}) \quad\text{for all } \bm{z}\in\Lambda,\; \bm{z}\neq \bm{y}.
\end{equation}
Thus $\bm{y}$ is the unique global minimizer of $F$ on $\Lambda$ and hence an interior
local minimizer of $F$ on $H$.

Applying the same Lagrange multiplier condition
\eqref{eq:lagrange}--\eqref{eq:coords-eq} under the constraint
\eqref{eq:constraint}, we again obtain
\begin{equation}\label{eq:y-uniform}
  \bm{y} = \Bigl(\frac{1}{n},\dots,\frac{1}{n}\Bigr) = \bm{u},
\end{equation}
which lies on the boundary of $\Delta_n^\downarrow$ by
\eqref{eq:uniform-boundary}. As before, \eqref{eq:ball-in-Lambda} forces
$\bm{y}$ to be an interior point of $\Delta_n^\downarrow$, yielding a contradiction.

\medskip

Hence neither the supremum nor the infimum of $\Lambda$ in the majorization
order can lie in the interior of $\Lambda$.
\end{proof}

The second result further assumes the compact set $\Lambda$ to be convex. 

\begin{proposition}
\label{thm:sup-inf-singular-final}
Let $\Lambda \subseteq \Delta_n^\downarrow$ ($n\ge 3$) be convex and compact with nonempty interior
with respect to the affine hyperplane $H$ defined in \eqref{eq:H-def}. Denote by $\bm{\lambda}^{\downarrow}_{\sup}$ and $\bm{\lambda}^{\downarrow}_{\inf}$ the supremum and infimum of $\Lambda$ in the majorization order. If $\bm{\lambda}^{\downarrow}_{\sup}$ or $\bm{\lambda}^{\downarrow}_{\inf}$ lies on the boundary of $\Lambda$, then it must be a singular boundary point. More precisely,
\begin{equation}\label{eq:singular-final}
\dim N_\Lambda(\bm{x}) = n-1 \ge 2
\quad\text{for }\bm{x}=\bm{\lambda}_{\sup}^\downarrow\text{ or }
\bm{\lambda}_{\inf}^\downarrow,
\end{equation}
where $N_\Lambda(\bm{x})$ denotes the outer normal cone of $\Lambda$ at $\bm{x}$ in $H$. 
\end{proposition}

\begin{proof}
We view $\Delta_n^\downarrow$ and $\Lambda$ as subsets of $H$, which is an
$(n-1)$-dimensional Euclidean space.
Let
\begin{equation}\label{eq:L-final}
L := \Bigl\{\bm{u}\in\mathbb{R}^n : \sum_{i=1}^n u_i = 0\Bigr\}
\end{equation}
denote the translation space of $H$, equipped with the Euclidean inner product.
For a closed convex set $C\subseteq H$ and $\bm{x}\in C$, the outer normal cone of
$C$ at $\bm{x}$ is defined by
\begin{equation}\label{eq:normal-cone-final}
N_C(\bm{x})
:= \{\bm{u}\in L : \bm{u}\cdot (\bm{y}-\bm{x}) \le 0,\ \forall\,\bm{y}\in C\}.
\end{equation}
A boundary point $\bm{x}\in\partial C$ is called \emph{singular} if
\begin{equation}\label{eq:singular-def-final}
\dim N_C(\bm{x})\ge 2.
\end{equation}

\medskip\noindent
\textbf{Preliminaries.}
Consider the linear function
\begin{equation}
f(\bm{x}) := \sum_{i=1}^n (n-i)\,x_i = \sum_{k=1}^{n-1} \sum_{i=1}^k x_i,\qquad \bm{x}\in H. \label{eq:f-partial-sums-final} 
\end{equation}
If $\bm{x},\bm{y}\in\Delta_n^\downarrow$ satisfy $\bm{x}\prec \bm{y}$, then
Eq.~(\ref{eq:majorization_1}) and Eq.~\eqref{eq:f-partial-sums-final} yield
\begin{equation}\label{eq:Schur-ineq-final}
f(\bm{x}) = \sum_{k=1}^{n-1} \sum_{i=1}^k x_i
    \le \sum_{k=1}^{n-1} \sum_{i=1}^k y_i
    = f(\bm{y}).
\end{equation}
Moreover, if $\bm{x}\ne \bm{y}$ and $\bm{x}\prec \bm{y}$, then at least one inequality
in (\ref{eq:majorization_1}) is strict. Hence
\begin{equation}\label{eq:strict-Schur-final}
\bm{x}\prec \bm{y},\ \bm{x}\ne \bm{y} \quad\Longrightarrow\quad f(\bm{x}) < f(\bm{y}).
\end{equation}
Thus $f$ is strictly Schur-convex on $\Delta_n^\downarrow$.

\smallskip
\noindent\textbf{Supremum.}
Define
\begin{equation}\label{eq:Lambda-prime-final}
\Lambda' := \{\bm{x}\in\Delta_n^\downarrow:\ \bm{x}\prec \bm{\lambda}_{\sup}^\downarrow\}.
\end{equation}
By definition of $\bm{\lambda}_{\sup}^\downarrow$ as the supremum of $\Lambda$ in
the majorization order,
\begin{equation}\label{eq:Lambda-subset-final}
\Lambda \subseteq \Lambda' \subseteq \Delta_n^\downarrow.
\end{equation}
The standard characterization of majorization gives
\begin{equation}\label{eq:maj-ineq-final}
\bm{x}\prec\bm{\lambda}_{\sup}^\downarrow
\iff
\sum_{i=1}^k x_i \le \sum_{i=1}^k (\bm{\lambda}_{\sup}^\downarrow)_i,\quad
k=1,\dots,n-1,
\end{equation}
together with $\sum_{i=1}^n x_i=\sum_{i=1}^n (\bm{\lambda}_{\sup}^\downarrow)_i=1$.
Combining \eqref{eq:maj-ineq-final} with the defining inequalities of
$\Delta_n^\downarrow$, we see that $\Lambda'$ is the intersection of finitely
many closed half-spaces in $H$ and is bounded. Hence $\Lambda'$ is a convex
polytope.

For any $\bm{x}\in\Lambda'$ with
$\bm{x}\ne\bm{\lambda}_{\sup}^\downarrow$, we have
$\bm{x}\prec\bm{\lambda}_{\sup}^\downarrow$ and hence
$f(\bm{x})<f(\bm{\lambda}_{\sup}^\downarrow)$.
Thus $\bm{\lambda}_{\sup}^\downarrow$ is the unique maximizer of the linear function $f$ on
$\Lambda'$, which implies that $\bm{\lambda}_{\sup}^\downarrow$ is an extreme point of $\Lambda'$:
\begin{equation}\label{eq:vertex-final}
\bm{\lambda}_{\sup}^\downarrow \text{ is a vertex of }\Lambda'.
\end{equation}

Since $\Delta_n^\downarrow$ has nonempty interior in $H$ and $\Lambda$ has
nonempty interior relative to $\Delta_n^\downarrow$, the set $\Lambda$
has nonempty interior in $H$. Thus $\operatorname{aff}(\Lambda)=H$, and
\eqref{eq:Lambda-subset-final} implies
\begin{equation}\label{eq:dim-Lambda-prime-final}
\operatorname{aff}(\Lambda') = H.
\end{equation}
Here $\operatorname{aff}(\Lambda)$ denotes the affine hull of $\Lambda$. Therefore, $\Lambda'$ is full-dimensional in $H$ with
\begin{equation}\label{eq:dim-H-final}
\dim \Lambda' = \dim H = n-1.
\end{equation}

For a full-dimensional polytope $P\subseteq H$ and a vertex $\bm{v}\in P$, the
normal cone $N_P(\bm{v})$ is generated by the outward normals of the facets
containing $\bm{v}$. In a full-dimensional polytope, at least $\dim H$ facets
meet at each vertex with linearly independent normals, so
\begin{equation}\label{eq:vertex-normal-final}
\dim N_P(\bm{v}) = \dim H.
\end{equation}
Applying this to $P=\Lambda'$ and $\bm{v}=\bm{\lambda}_{\sup}^\downarrow$ with
\eqref{eq:dim-H-final} yields
\begin{equation}\label{eq:normal-Lambda-prime-final}
\dim N_{\Lambda'}(\bm{\lambda}_{\sup}^\downarrow) = n-1.
\end{equation}

If $C\subseteq D$ are closed convex sets in $H$ and $\bm{x}\in C$,
\eqref{eq:normal-cone-final} implies
\begin{equation}\label{eq:cone-inclusion-final}
N_D(\bm{x})\subseteq N_C(\bm{x}),
\end{equation}
since any vector supporting the larger set $D$ at $\bm{x}$ also supports $C$.
From \eqref{eq:Lambda-subset-final} and \eqref{eq:cone-inclusion-final},
\begin{equation}\label{eq:cone-inclusion-Lambda-final}
N_{\Lambda'}(\bm{\lambda}_{\sup}^\downarrow)
\subseteq N_{\Lambda}(\bm{\lambda}_{\sup}^\downarrow),
\end{equation}
whence
\begin{equation}\label{eq:normal-Lambda-dim-final}
\dim N_\Lambda(\bm{\lambda}_{\sup}^\downarrow)
\ge \dim N_{\Lambda'}(\bm{\lambda}_{\sup}^\downarrow)
= n-1.
\end{equation}
Since the normal cone is defined with respect to $H$,
\begin{equation}\label{eq:normal-Lambda-dim-upper-bound}
\dim N_\Lambda(\bm{\lambda}_{\sup}^\downarrow) \leq \dim L = \dim H = n-1.
\end{equation}
Combining \eqref{eq:normal-Lambda-dim-final} and \eqref{eq:normal-Lambda-dim-upper-bound}, we obtain
\begin{equation}
\dim N_\Lambda(\bm{\lambda}_{\sup}^\downarrow)
= n-1 \ge 2.    
\end{equation}
This establishes that $\bm{\lambda}_{\sup}^\downarrow$ is a singular boundary point of $\Lambda$.

\smallskip
\noindent\textbf{Infimum.}
The argument for $\bm{\lambda}_{\inf}^\downarrow$ is analogous.
Define
\begin{equation}\label{eq:Lambda-tilde-final}
\widetilde{\Lambda} := \{\bm{x}\in\Delta_n^\downarrow:\ \bm{\lambda}_{\inf}^\downarrow \prec \bm{x}\}.
\end{equation}
Then $\Lambda\subset\widetilde{\Lambda}\subset\Delta_n^\downarrow$. The
same reasoning shows that $\widetilde{\Lambda}$ is a full-dimensional
polytope in $H$.

Since $\bm{\lambda}_{\inf}^\downarrow \prec \bm{x}$ for all $\bm{x}\in\widetilde{\Lambda}$
and $f$ is strictly Schur-convex,
\begin{equation}\label{eq:inf-minimizer-final}
\bm{x}\in\widetilde{\Lambda},\ \bm{x}\ne\bm{\lambda}_{\inf}^\downarrow
\quad\Longrightarrow\quad
f(\bm{\lambda}_{\inf}^\downarrow) < f(\bm{x}).
\end{equation}
Hence $\bm{\lambda}_{\inf}^\downarrow$ is the unique minimizer of $f$ on
$\widetilde{\Lambda}$ and therefore a vertex of $\widetilde{\Lambda}$.
As above,
\begin{equation}\label{eq:normal-Lambda-tilde-final}
\dim N_{\widetilde{\Lambda}}(\bm{\lambda}_{\inf}^\downarrow)=n-1.
\end{equation}
Using $\Lambda\subset\widetilde{\Lambda}$ and
\eqref{eq:cone-inclusion-final},
\begin{equation}\label{eq:normal-Lambda-inf-final}
\dim N_\Lambda(\bm{\lambda}_{\inf}^\downarrow)
\ge \dim N_{\widetilde{\Lambda}}(\bm{\lambda}_{\inf}^\downarrow)
= n-1\ge 2.
\end{equation}
We also have
\begin{equation}\label{eq:normal-Lambda-inf-upperbound}
\dim N_\Lambda(\bm{\lambda}_{\inf}^\downarrow) \leq \dim L = \dim H = n-1.    
\end{equation}
Combining \eqref{eq:normal-Lambda-inf-final} and \eqref{eq:normal-Lambda-inf-upperbound}, we obtain
\begin{equation}
\dim N_\Lambda(\bm{\lambda}_{\inf}^\downarrow)
= n-1 \ge 2.    
\end{equation}
Hence $\bm{\lambda}_{\inf}^\downarrow$ is also a singular boundary point of $\Lambda$.
\end{proof}

\begin{corollary*}
Let $\Lambda \subseteq \Delta_n^\downarrow$ ($n\ge 3$) be convex and compact with nonempty interior. Denote by $\bm{\lambda}^{\downarrow}_{\sup}$ and $\bm{\lambda}^{\downarrow}_{\inf}$ the supremum and infimum of $\Lambda$ in the majorization order. If $\Lambda$ is smooth, then both $\bm{\lambda}^{\downarrow}_{\sup}$ and $\bm{\lambda}^{\downarrow}_{\inf}$ lie outside $\Lambda$.   
\end{corollary*}

\end{document}